\RequirePackage[l2tabu, orthodox]{nag}
\documentclass[11pt, onecolumn]{IEEEtran}

\usepackage{amssymb}
\usepackage{amsfonts}
\usepackage{epsfig}
\usepackage{subfigure}
\usepackage{calc}
\usepackage{color}
\usepackage{bm}
\usepackage{enumerate}
\usepackage{pdfsync}
\usepackage{framed}
\usepackage{arydshln}
\usepackage[fleqn]{amsmath}
\usepackage[nosort]{cleveref}
\usepackage{natbib}
\setcitestyle{numbers,square,comma}

\newtheorem{defn}{Definition}

\newtheorem{thm}{Theorem}[section]
\newtheorem{cor}[thm]{Corollary}

\newtheorem{lem}[thm]{Lemma}

\newtheorem{constr}[thm]{Construction}

\newcommand{\bit}{\begin{itemize}}
\newcommand{\eit}{\end{itemize}}
\newcommand{\bcor}{\begin{cor}}
\newcommand{\ecor}{\end{cor}}

\begin{document}

\title{The Storage vs Repair-Bandwidth Trade-off for Clustered Storage Systems}
\author{N. Prakash, Vitaly Abdrashitov and Muriel M{\'{e}}dard
	\thanks{N. Prakash, Vitaly Abdrashitov and Muriel M{\'{e}}dard are with the Research Laboratory of Electronics, Massachusetts Institute of Technology, USA (email: \{prakashn, vit, medard\}@mit.edu). A part of the results in the paper were presented as part of an invited talk at the 54th Annual Allerton Conference on Communication, Control, and Computing, 2016, Allerton Park and Retreat Center, Monticello, IL, USA. This work is in part supported by the Air Force Office of Scientific Research (AFOSR) under award No FA9550-14-1-043, and in part supported by the National Science Foundation (NSF) under Grant No.  CCF-1527270.}
}

\maketitle

\begin{abstract}
We study a generalization of the setting of regenerating codes, motivated by applications to storage systems consisting of  clusters of storage nodes. There are $n$ clusters in total, with $m$ nodes per cluster.  A data file is coded and stored across the $mn$ nodes, with each node storing $\alpha$ symbols. For availability of data, we require that the file be retrievable by downloading the entire content from any subset of $k$ clusters. Nodes represent entities that can fail. We distinguish between intra-cluster and inter-cluster bandwidth (BW) costs during node repair.  Node-repair in a cluster is accomplished by downloading $\beta$ symbols each from any set of $d$ other clusters, dubbed remote helper clusters, and also up to $\alpha$ symbols each from any set of $\ell$ surviving nodes, dubbed local helper nodes, in the host cluster. We first identify the optimal trade-off between storage-overhead and inter-cluster repair-bandwidth under functional repair, and also present optimal exact-repair code constructions for a class of parameters. The new trade-off is strictly better than what is achievable via space-sharing existing coding solutions, whenever $\ell > 0$. We then obtain sharp lower bounds on the necessary intra-cluster repair BW to achieve optimal trade-off. Under functional repair, random linear network codes (RLNCs) simultaneously optimize usage of both inter- and intra-cluster repair BW; simulation results based on RLNCs suggest optimality of the bounds on intra-cluster repair-bandwidth. Our bounds reveal the interesting fact that, while it is beneficial to increase the number of local helper nodes $\ell$ in order to improve the storage-vs-inter-cluster-repair-BW trade-off, increasing $\ell$ not only increases intra-cluster BW in the host-cluster, but also increases the intra-cluster BW in the remote helper clusters.  We also analyze resilience of the clustered storage system against passive eavesdropping by providing file-size bounds and optimal code constructions.
\end{abstract}


\section{Introduction}

We consider the problem of designing efficient erasure codes for fault-tolerant data storage in a clustered network of storage nodes. Nodes within a cluster are connected to each other via one network, while a second network provides connectivity between clusters.  In clustered networks, there is often a differentiation between intra- and inter-cluster bandwidth costs, and this occurs  because of factors like physical distance between the clusters or some other differentiating characteristic in the communications within and between clusters.  Typically, intra-cluster bandwidth cost is much less than inter-cluster bandwidth cost. A user file is erasure coded and stored across the nodes in the various clusters. From an availability perspective, it is of interest to encode and store a data file such that access to a subset of clusters allows reconstruction of the entire uncoded file. We consider nodes as failure domains, and require efficient repair of failed nodes in any cluster. Repair of a failed node in a cluster is performed by downloading \emph{helper data} from other clusters, and also a subset of surviving nodes in the host-cluster. We permit the possibility that a cluster that aids in the repair of a failed node gathers and processes helper-data from its various nodes, before sending it to the target-node, in order to decrease inter-cluster repair-bandwidth costs. A good erasure-code solution allows a desirable trade-off between storage-overhead and the inter- and intra-cluster repair BW costs for a given availability requirement. 

The model described above is motivated by applications to cloud storage settings, where user data is spread across distinct data-centers. For instance, the clusters could represent geographically separated data centers of a cloud-service provider or a content-delivery network. All major cloud service providers like Amazon \cite{amazonwhitepaper}, Microsoft Azure \cite{azurestorage} and Akamai \cite{akamai} provide options for geo-replicating user data in multiple data centers. Geo-replication across data centers improves both availability and durability of user data. Data-center unavailability occurs owing to network or power outages, software bugs and even security vulnerabilities. Note that in the event of unavailability, data is never lost and thus does not necessitate data-center wide rebuild. While the three systems mentioned above do not mention erasure coding across data centers, other practical systems indeed consider this possibility.  For example, Facebook's F$4$ blob storage system \cite{facebookf4} considers storing the xored content of two data centers  in a third data-center. Similarly, the authors of the Hitachi white-paper \cite{hitachi} suggest using Reed-Solomon like codes across data centers for  private clouds consisting of a small number of data-centers. Microsoft Giza~\cite{giza}, a geo-distributed storage system considers erasure coding data across as many as $11$ data centers spread over $3$ continents. Giza employs MDS codes like Reed-Solomon codes across data centers to handle availability requirement during data collection; i.e., a user connects to the nearest $k$ available data centers (assuming an $[n, k]$ MDS code is used) during data collection to download the data. The paper notes that (we quote) ``cross-DC erasure coding only becomes economically attractive if 1) there are workloads that consume very large storage capacity while incurring very
little cross-DC traffic; 2) there are enough cross-DC network bandwidth at very low cost". The paper provides justification for both these aspects;  we refer to \cite{giza} for details.

Our model is perhaps even more applicable to the setting of ``cloud of clouds", where user data is spread across data centers corresponding to multiple cloud service providers. Implementation studies that show the advantages of using Reed-Solomon-like erasure codes to store data in user-defined cloud-of-clouds appear  in \cite{racs, mulibcloud, depsky, cyrus}. Several reasons motivate a cloud-of-cloud setting, instead of a single cloud setting. The first is the need for high availability in a local geographic zone~\cite{mulibcloud} - individual cloud service providers might be limited in their number of geographic zone specific data centers. Another motivation for a cloud-of-cloud setting is the flexibility to avoid vendor lock-in~\cite{racs, mulibcloud}. With an ever increasing number of new cloud providers offering competitive pricing and features,  it might be of interest to migrate from an existing cloud provider to a new one. However, individual cloud providers charge the users differently for in-network (intra-cluster in our model) and out-of-network (inter-cluster) data movement. In this scenario, to decrease the cost of migration, it is beneficial to spread data across several providers, so that the user only needs to migrate data in the less competitive providers to the new provider. Finally, a reason for providing user-defined cloud-of-clouds is that of privacy concerns~\cite{depsky, cyrus, Olivera}, where security compromise of any one cloud provider does not compromise  user data.

In this work, we model an entire data-center as a cluster. We restrict ourselves to the storage of a single user file; the solution provided here can be applied independently to any of the user files. For storage of a single coded file, we assume an equal number of storage nodes in any of the clusters. While performing data collection, we assume a cluster to be either completely available or completely unavailable. Such an assumption is sensible in a multi-data-center cloud setting~\cite{hitachi}. In our model, we restrict ourselves to recovery from single node failure. If there are multiple node failures in the system, it is assumed that the recovery process happens sequentially, one node at a time. While catastrophic failure of an entire data-center is rare, correlated failures of nodes in a data center are an important issue reported in practice \cite{ford2010availability}. In this scenario, in our model, we parametrize the number of nodes, named local helper nodes, in the host-cluster that can aid in the repair of a node in the host-cluster.  Further, to keep the model simple, we ignore any hierarchical topologies that may be present inside a data-center (cluster), and simply assume equal cost connectivity between any two nodes inside a cluster. We also assume direct connectivity between any two clusters in the network.

A straightforward solution that guarantees availability of data, and minimizes inter-cluster repair-bandwidth costs is simply to use a product code consisting of two Maximum Distance Separable (MDS) codes - one across the clusters, and another within a cluster. While the solution entirely eliminates any inter-cluster bandwidth cost, it suffers from poor storage-overhead owing to the need to have redundancy in every cluster. At the other extreme, it is also possible to achieve highly optimized storage cost, at the expense of inter-cluster repair-bandwidth, by using minimum storage regenerating (MSR) codes \cite{dimakis} across the clusters. For instance, if we assume there are $n$ clusters in total, then  $(n, k)$ MSR codes across the clusters ensure that data is retrievable by accessing (entire) content of any $k$ clusters. However, the solution suffers from high inter-cluster repair-bandwidth cost. Our goal in this work, at a high level, is to explore alternate solutions which can smoothly trade-off storage-overhead against inter-cluster repair-bandwidth for given availability requirements. We show that it is indeed possible to achieve (see Fig. \ref{fig:spaceshare}) operating points which are strictly better than those obtained by space-sharing between product MDS codes and MSR codes. We also characterize the amount of intra-cluster repair-bandwidth needed for achieving the optimal trade-off between storage-overhead and inter-cluster repair bandwidth.

For the rest of the introduction, we first provide an abstract description of the system model used in this work. This is followed by a discussion of other related system models in the literature, a summary of our results, and also an example of the proposed erasure-coding solution.

\subsection{System Model} \label{sec:sys_model}

\begin{figure*}
	\centering
	\subfigure[Data Collection]{\label{fig:data_collec}\includegraphics[height=2.5in]{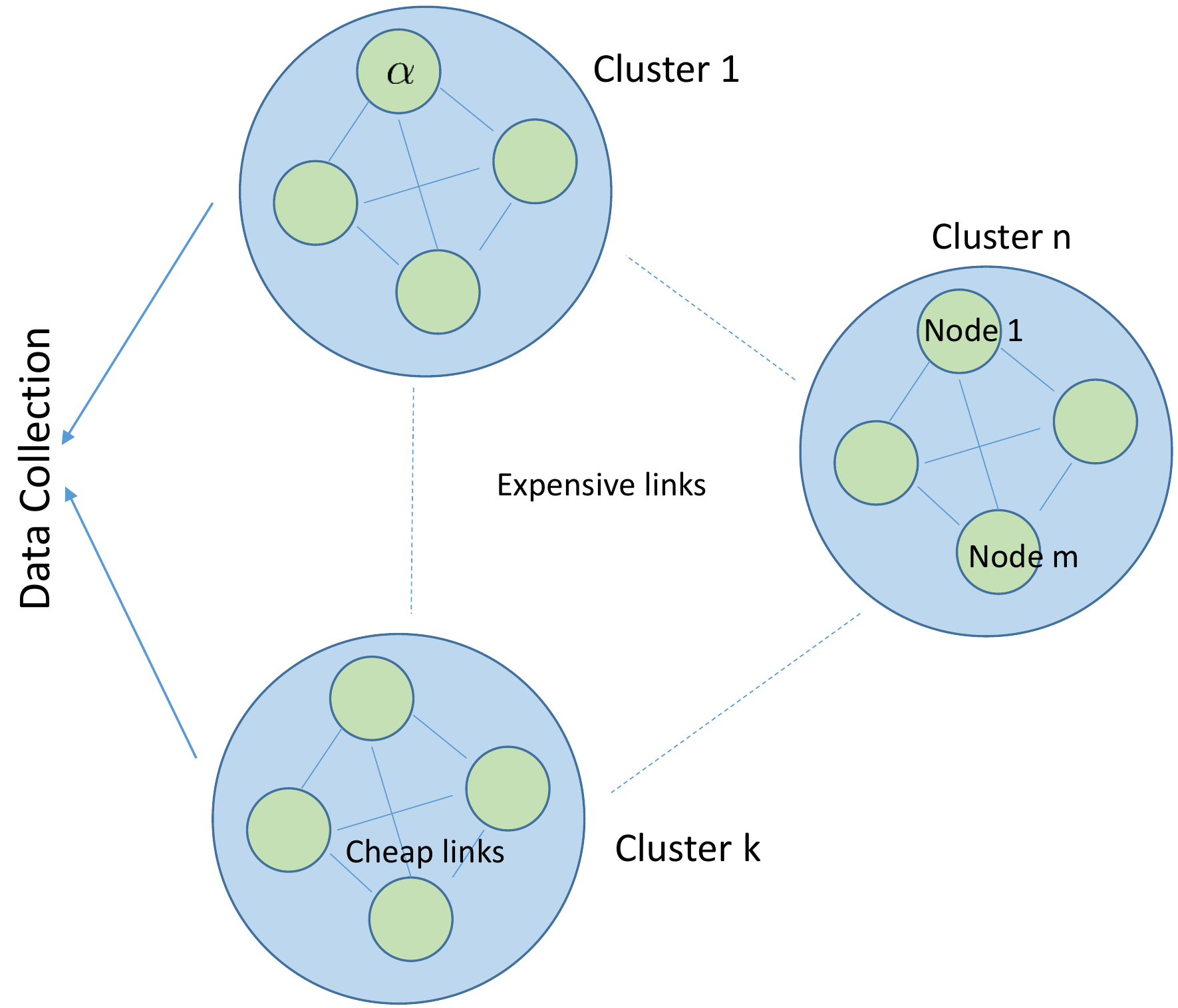}}
	\hspace{0.05in}
	\subfigure[Node Repair]{\label{fig:node_repair}\includegraphics[height=2.5in]{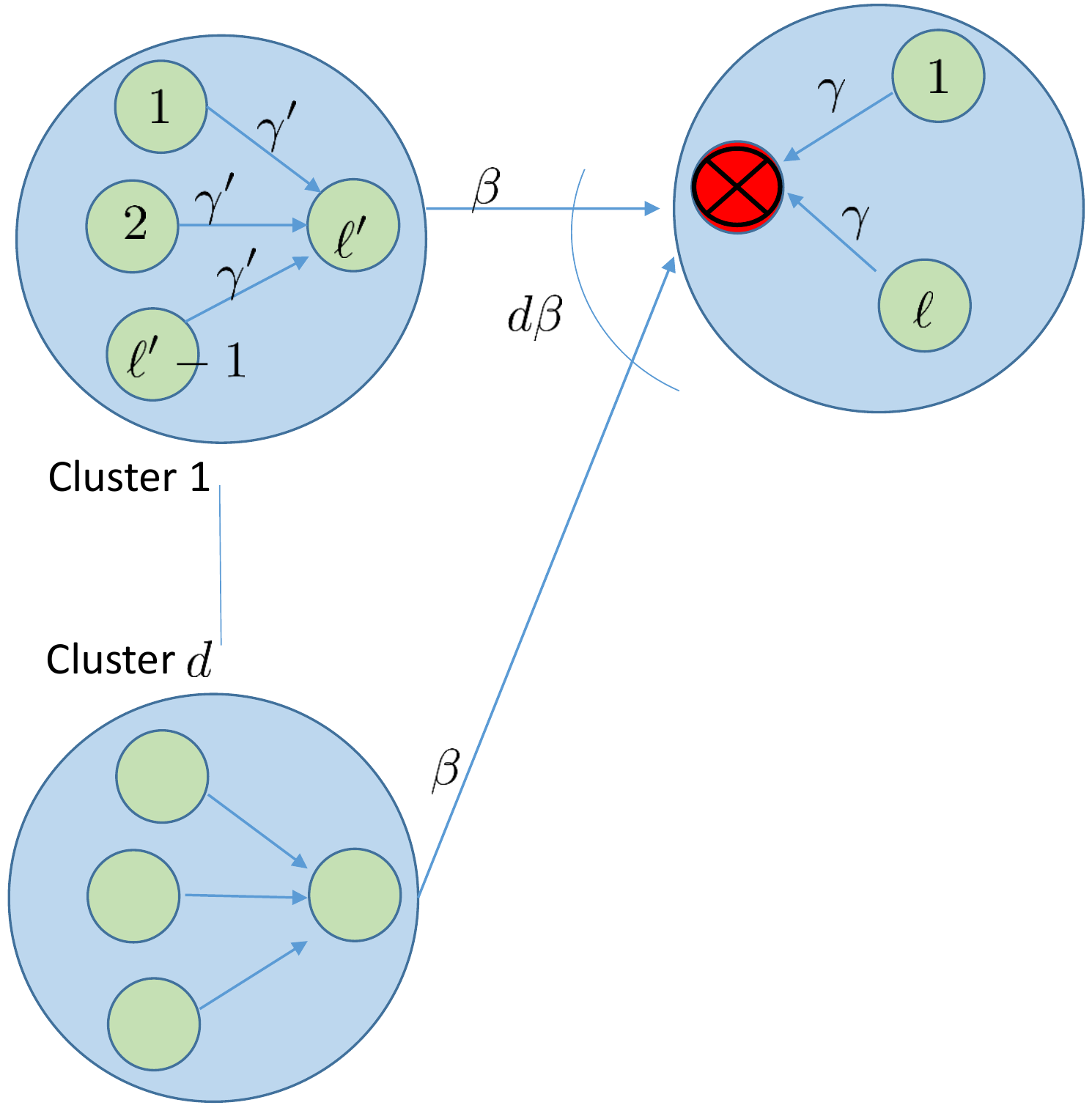}}
	\caption{System model for data collection and Node repair in a clustered storage network. Data Collection is accomplished by downloading entire contents of any $k$ clusters. Repair of a failed node in a cluster is accomplished by contacting $(i)$ any $\ell$ surviving nodes in the host-cluster, and downloading all their content, and $(ii)$ any $d$ other clusters, and downloading $\beta$ symbols from each of them.}
	\label{fig:sys_model}
\end{figure*}

We propose a natural generalization of the setting of regenerating codes (RC)~\cite{dimakis} for clustered storage networks. The network consists of $n$ clusters, with $m$ nodes in each cluster. The network is fully connected such that any two nodes within a cluster are connected via an intra-cluster link, and any two clusters are connected via an inter-cluster link. A node in one cluster that needs to communicate with another node in a second cluster does so via the corresponding inter-cluster link.  A data file of size $B$ symbols is encoded into $nm\alpha$ symbols, and stored across the $nm$ nodes such that each node stores $\alpha$ symbols. The symbols are assumed to come from a finite field $\mathbb{F}_q$ of $q$ elements. For data collection we have an availability constraint such that the entire content of any $k$ clusters be sufficient to recover the original data file~(Fig. \ref{fig:sys_model}). As mentioned before, nodes represent failure domains, and we restrict ourselves to the case of efficient recovery from single node failure. Node repair is parametrized by three parameters $d, \beta$ and $\ell$. We assume that the replacement of a failed node is in the same cluster as the failed node. The replacement node downloads $\beta$ symbols each from any set of $d$ other clusters, dubbed remote helper clusters. The $\beta$ symbols from any of the remote helper clusters are possibly a function of the $m\alpha$ symbols  present in the cluster - we assume that any one of the nodes in the cluster takes responsibility for computing these $\beta$ symbols before passing them outside the cluster\footnote{One could also assume the presence of a dedicated compute unit to compute the helper data. Such an assumption allows us to enforce a symmetric demand on the usage of intra-cluster bandwidth for all the nodes in the remote helper cluster. We will rely on the presence of such dedicated compute units in the information flow graphs (see Section \ref{sec:IFG})  used to derive file-size upper bounds.}. Further, we also permit the replaced node to download (entire) content from any set of $\ell$ other nodes, dubbed local helper nodes, in the host cluster, during the repair process. The quantity $d\beta$ represents the inter-cluster  repair-bandwidth. We refer to the overall code as the \emph{generalized regenerating code} (GRC) $\mathcal{C}_m$ with parameters $\{(n, k, d) (\alpha, \beta) (m, \ell)\}$.

The model reduces to the setup of RCs in \cite{dimakis}, when $m = 1$ (in which case, $\ell = 0$ automatically). We shall refer to the setup in \cite{dimakis} as the classical setup or classical regenerating codes. Our generalization has two additional parameters $\ell$ and $m$ when compared with the classical setup. As in the classical setup we consider both the notions of functional and exact repair. Under exact repair, the content of the repaired node is identical to that of the failed node; while in functional repair, the repair content permits data collection and repair of additional failed nodes. The first goal of the paper is to obtain a trade-off between storage-overhead $nm\alpha / B$ and inter-cluster repair-bandwidth $d\beta$ for an $\{(n, k, d) (\alpha, \beta) (m, \ell)\}$ GRC. We further note that, unlike the classical setup, the generalized setup permits $d < k$. In our model, whenever $d > 0$, we assume that the encoding function does not introduce any local dependence among the content of the various nodes of a cluster\footnote{Under linear encoding, the coded content of cluster $i, 1 \leq i \leq n$ can be written as $\widehat{\bf m}G_i$, where $\widehat{\bf m}$ is the message vector of length $B$, and $G_i$ is a $B \times m\alpha$ matrix. In this case, when we say that the encoding not introduce local dependence, we mean that the matrix $G_i$ has full column-rank}; for example, the model excludes the possibility of a local parity node within a cluster, which would hold the component-wise sum in $\mathbb{F}_q^{\alpha}$ of the other nodes' data. As we shall show, the case $d = 0$ is a special one where local dependence is necessary.

The model described above does not consider intra-cluster bandwidth incurred during repair. Intra-cluster bandwidth is needed, firstly, to compute the $\beta$ symbols in any remote helper cluster, and, secondly, to download content from $\ell$ local helper nodes in the host cluster. In order to characterize the amount of intra-cluster bandwidth that is needed to establish optimal trade-off between storage-overhead and inter-cluster repair-bandwidth, we consider the repair model shown in Fig. \ref{fig:node_repair}. In this model, the replacement node downloads at most $\gamma, \gamma \leq \alpha$ symbols  from each of the $\ell$ local helper nodes from the host-cluster. With regard to a remote helper cluster, we assume that the $\beta$ symbols contributed by it are only a function of at most $\ell', \ell' \leq m$ nodes of the cluster. We make the assumption that any set of $\ell'$ nodes can be used to compute the $\beta$ symbols. Further, we limit the amount of data that each of these $\ell'$ nodes can contribute to at most $\gamma', \gamma' \leq \alpha$ symbols. A second goal of this paper is to identify necessary requirements on the parameters $\gamma, \ell', \gamma'$ that are needed to guarantee optimal trade-off between storage-overhead and inter-cluster repair-bandwidth.

A summary of the various parameters used in the description of the system model appears in Table \ref{tab:parameters}.

\begin{table}[tb]
	\caption{Notation used in our System Model. The last three parameters relate to intra-cluster repair-bandwidth}
	\label{tab:parameters}
	\centering

	\begin{tabular}{cl}
		\textbf{Symbol} & \textbf{Definition}\\
		\hline
		$n$ & total number of clusters in the system\\
		$m$ & number of storage nodes in each cluster\\
		$k$ & number of clusters required for data collection \\
		$d$ & number of remote helper clusters providing helper data during node repair\\
		$\ell$ & number of local helper nodes providing helper data during node repair \\
		$q$ & finite field size for data symbols \\
		$\alpha$ & number of symbols each storage node holds for one coded file \\
		$\beta$ & size of helper data downloaded from each remote helper cluster during  node repair, in symbols \\ \hdashline
		$\gamma$ & size of helper data downloaded from each local helper during  node repair, in symbols \\
		$\ell'$ & number of nodes in a remote helper cluster that contribute toward computing the $\beta$ helper symbols of the cluster \\
		$\gamma'$ & size of the data provided by each of $\ell'$ nodes at remote helper cluster to compute the cluster helper data, in symbols \\
		\hline
	\end{tabular}
\end{table}

\subsection{Related Work}

Regenerating codes were originally introduced in \cite{dimakis} for simultaneously optimizing storage overhead and repair bandwidth for flat storage systems. By a flat storage system, it is meant that every node in the storage system is connected to every other storage node via some logical link, where all logical links incur the same bandwidth cost for communication per bit. Further, the data collector also connects to any of the storage nodes via links of similar cost.  There has since then been significant progress in the area of classical RCs in terms of code constructions, finding optimal trade-offs under exact repair, and practical implementation. Below, we review variations of RCs that have been proposed for non-flat topologies, and see how our model relates to these existing variations. We shall also comment on how the model of locally repairable codes~\cite{gopalan2012locality} relates to our model.

\subsubsection{Regenerating Code Variations for Clustered Topologies}

Regenerating code variations for data-center like topologies consisting of racks and nodes are considered in \cite{plee_isit2016_doubleregen, hu2017optimal, clust_stor_Moon, gaston2013realistic, Gaston_nonhom, ozan_xyregen}. In  \cite{plee_isit2016_doubleregen}, \cite{clust_stor_Moon} and \cite{gaston2013realistic}, the authors distinguish between inter-rack (inter-cluster) and intra-rack (intra-cluster) bandwidth costs. Further, the works \cite{plee_isit2016_doubleregen} and  \cite{clust_stor_Moon}  permit pooling of intra-rack helper data to decrease inter-rack bandwidth. Also, all three works allow taking help from host-rack nodes during repair. Unlike our model, for data collection all three works simply require file retrievability from any set of $k$ nodes irrespective of the racks (clusters) to which they belong. In other words, the notion of clustering applies only to repair, and not data collection, and this is a major difference with respect to our model. Thus while these variations are suitable for modeling the node-rack topologies present within a data center, they do not model situation  of erasure coding across data centers with the availability requirement as considered in this work. The work \cite{hu2017optimal} applies the theoretical results of \cite{plee_isit2016_doubleregen} for the practical setting of Hadoop file system. The work in \cite{Gaston_nonhom} is a variation of that in \cite{gaston2013realistic} for a two-rack model, where the per-node storage capacity of the two racks differ. In \cite{ozan_xyregen}, the authors consider a two-layer storage setting like ours, consisting of several blocks (analogous to clusters as considered in this work) of storage nodes. A different clustering approach is followed for both data collection and node repair. For data collection, one accesses $k_c$ nodes each from any of $b_c$ blocks. Though \cite{ozan_xyregen} focuses on node repair, the model assumes possible unavailability of the whole block where the failed node resides, and as such uses only nodes from other blocks for repair. Further, unlike our model in this work, the authors do not differentiate between  inter-block and intra-block bandwidth costs. The framework of twin-codes introduced in \cite{twincode} is also related to our model and implicitly contains the notion of clustering. In \cite{twincode} nodes are divided into two sets. For data collection, one connects to any $k$ nodes in the same set. Recovery of a failed node in one set is accomplished by connecting to $d$ nodes in the other set. However, there is no distinction between intra-set and inter-set bandwidth costs, and this becomes the main difference with our model.

\subsubsection{Regenerating Code Variations for Heterogeneous Systems}

Several works \cite{nihar_flexible, Shum_hetero, ernvall2013capacity, akhlaghi2010cost, li2010tree, hetero_regen_time} study variations of RCs in varied settings, with different combinations of node capacities,  link costs, and amount of data-download-per-node. The main difference between our model and these works is that none of them explicitly considers clustering of nodes while performing data collection. In \cite{nihar_flexible}, the authors introduce flexible regenerating codes for a flat topology of storage nodes, where uniformity of download is enforced neither during data collection, nor during node-repair. References \cite{Shum_hetero}, \cite{ernvall2013capacity} consider systems where the storage and repair-download costs are non-uniform across the various nodes. The authors of \cite{Shum_hetero}, as in \cite{nihar_flexible}, allow  a replacement node to download an arbitrary amount of data from each helper node. In  \cite{akhlaghi2010cost}, nodes are divided into two sets, based on the cost incurred while these nodes aid during repair. As noted in \cite{Gaston_nonhom}, the repair model of \cite{akhlaghi2010cost} is different from a clustered network, where the repair cost incurred by a specific helper node depends on which cluster the replacement node belongs to. The works of \cite{li2010tree} and \cite{hetero_regen_time} focus on minimizing regeneration time rather regeneration bandwidth in systems with non-uniform point-to-point link capacities. Essentially, each helper node is expected to find the optimal path, perhaps via other intermediate nodes, to the replacement node such that the various link capacities are used in a way to transfer all the helper data needed for repair in the shortest possible time. It is interesting to note both of these works permit pooling of data at an intermediate node, which gathers and processes any relayed data with its own helper data. Recall that our model (and the one in \cite{plee_isit2016_doubleregen}) also considers pooling of data within a remote helper cluster, before passing on to the target cluster.

\subsubsection{Locally Repairable Codes}

Locally repairable codes (LRCs), introduced in \cite{gopalan2012locality}, are motivated by the need to carry out efficient node repair in hierarchical storage systems, such as data centers. The subject of LRCs, like regenerating codes, has attained significant attention both in theory and practice, since its introduction. Recovery of a node failure within a cluster is first attempted locally, in order to minimize cross-cluster bandwidth; if too many nodes in the cluster are unavailable, global parity nodes, spread across various clusters, aid in recovery. However, LRCs do not model clustering of nodes while carrying out data collection, and this is once again a major difference with our model. To draw further similarities with LRCs, we note that in our model, during the repair of a node, the helper data from the $\ell$ nodes from the same cluster can be considered as \emph{local-helper data}. However if $d >0$, unlike in LRCs, local-helper data alone is not sufficient for single-node-repair, since in this work we assume the encoding function of the generalized regenerating code to introduce no dependency among the content of the $m$ nodes in any cluster. For $d > 0$, lack of dependencies within each cluster reduces storage overhead when compared to LRCs; at the same time,
local-helper data allows us to achieve a trade-off between storage-overhead and the inter-cluster-repair-bandwidth better than that of a system with $m$ stacked classical RCs, in which the set of $i$-th nodes in all clusters stores the data encoded with $i$-th RC.
The case $d = 0$ corresponds to one having no remote helper clusters, so that the repair is carried out entirely locally. The difference in this case of $d = 0$ with model of LRCs is notion of clustering of nodes, while performing data collection.

A tabular summary of related works appears in Table \ref{tab:relatedworks}.

\renewcommand{\arraystretch}{1.5}
\begin{table}[tb]
	\caption{Summary of related works, and key differences from generalized regenerating codes (GRC)}
	\label{tab:relatedworks}
	\centering

	\begin{tabular}{p{3cm}p{13.5cm}}
		\textbf{Reference} & \textbf{Focus}\\
		\hline

		\multicolumn{2}{c}{RCs in clustered topologies} \\
		\citet{plee_isit2016_doubleregen,clust_stor_Moon,gaston2013realistic} & Rack-based topology with cluster-based repair, allow use of cheaper local-rack helper data. \cite{plee_isit2016_doubleregen} and  \cite{clust_stor_Moon} permit pooling of intra-rack helper data to decrease inter-rack bandwidth. Unlike GRC, the works assume node-based data collection with no notion of clustering for data collection \\
		\citet{Gaston_nonhom} & Similar to \citet{gaston2013realistic} with 2 racks with different per-node storage \\
		\citet{ozan_xyregen} & Two-layered storage with nodes grouped in blocks, block failure model, repair and data collection from fixed number of nodes from surviving blocks. All links have same costs \\
		\citet{twincode} & Two sets of nodes, repair from $d$ nodes of the other set, data collection from $k$ nodes of any single set. All links have same costs\\ \hline
		\multicolumn{2}{c}{RCs in heterogeneous systems} \\
		\citet{nihar_flexible} & Non-uniform amount of data downloaded from nodes during repair or data collection, without any notion of clustering \\
		\citet{Shum_hetero,ernvall2013capacity} & Non-uniform storage sizes and repair BW costs \\
		\citet{akhlaghi2010cost} & Repair BW cost can take two different values, depending on the helper node used, this is different from clustering \\
		\citet{li2010tree,hetero_regen_time} & Aim to transmit the repair helper data in shortest amount of time by finding the optimal path from helper nodes to replacement nodes through a network with non-uniform link capacities \\ \hline
		\citet{gopalan2012locality} & Inherent notion of clustering for node repair, and single node is performed locally with the help of other nodes in the clusters. Thus unlike GRC, LRCs assume dependency among nodes in a cluster for single node repair. Further, there is no notion of clustering for data collection. \\ \hline
	\end{tabular}
\end{table}
\renewcommand{\arraystretch}{1.0}

\subsection{Our Results}

\subsubsection{Upper Bound on File Size $B$}

Under the setting of functional repair, the file-size $B$ is shown to be upper bounded by
\begin{eqnarray}\label{eq:file_size}
B & \leq & B^* \ = \ \ell k\alpha + (m - \ell)\sum_{i=0}^{k-1} \min \{\alpha, (d-i)^+\beta \},
\end{eqnarray}
where we use the notation $a^+$ to mean $\max(a, 0)$, for any integer $a$. The bound is shown by considering the \emph{information-flow graph} (see Section \ref{sec:IFG}) under functional-repair, and calculating the minimum cut.  For any finite information-flow graph, the achievability of \eqref{eq:file_size} follows from results in network coding~\cite{KoetterMedard}. This establishes the optimality of \eqref{eq:file_size} when we know an upper bound on the number of node failures that occur during the life-time of the system. In practice, random linear network codes (RLNCs)~\cite{rlnc} can be used to achieve near-optimal operating points.

For fixed values of $B = B^*, n, k, d>0, \ell, m$, \eqref{eq:file_size} gives a normalized trade-off (see Fig. \ref{fig:tradeoff}) between storage-overhead $nm\alpha/B$ and inter-cluster-repair-bandwidth-overhead $d\beta/\alpha$ (see Section \ref{sec:tradeoff} as to why we refer to $d\beta/\alpha$ as repair-bandwidth-overhead). For any $m$, when $\ell = 0$, the trade-off is exactly same as that of the classical regenerating codes~\cite{dimakis}. When $\ell > 0$ (implies $m > 1$), the  trade-off is strictly better than that of the classical setup.

\begin{figure}
	\centering
	\includegraphics[width=95mm]{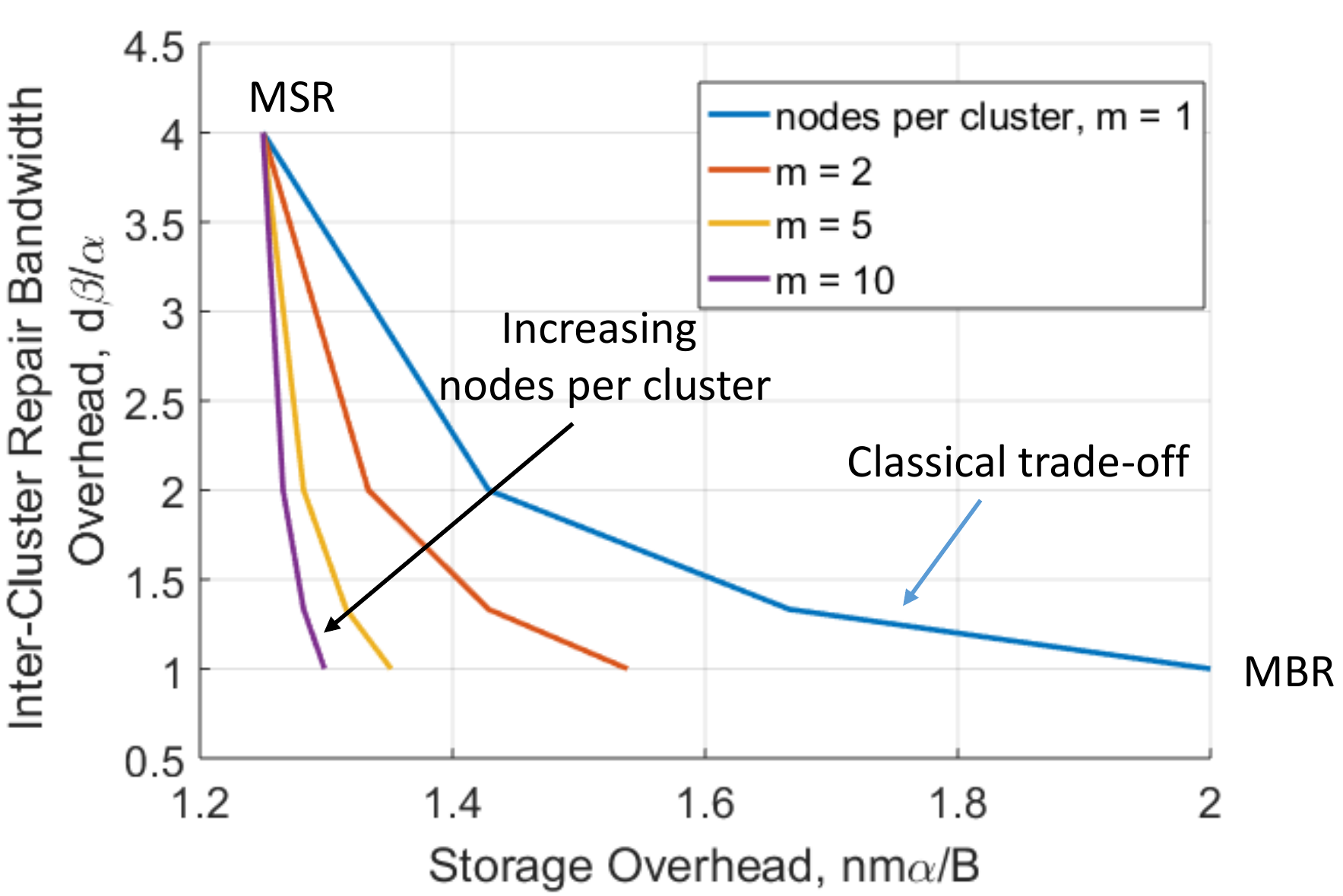}
	\caption{Trade-off between storage-overhead $nm\alpha/B$ and inter-cluster repair-bandwidth-overhead $d\beta/\alpha$, for an $(n = 5, k = 4, d = 4)$ clustered storage system, with $\ell = m-1$. }
	\label{fig:tradeoff}
\end{figure}

\vspace{0.1in}

\subsubsection{Optimal Code Constructions}

We present optimal codes for the minimum storage-overhead  and the minimum inter-cluster  repair-bandwidth-overhead  operating points of the trade-off, for a class of parameters and under the setting of exact repair. These two operating points corresponding to MSR and MBR codes, respectively. The operating points of the MSR and MBR  codes are characterized by relations $B = mk\alpha$ and $\alpha = d\beta$ respectively. The construction works by suitably combining $\ell$ $[n, k]$ vector MDS codes over $\mathbb{F}_q^{\alpha}$ and $(m - \ell)$ classical exact-repair MSR (MBR) codes. We also present an optimal code construction for functional-repair, which  tolerates an arbitrary number of failures\footnote{The network-coding based achievability works only if there is a known upper bound on the number of repairs that occur over the duration of operation of the system.}, for the case $\ell = m-1, d \geq  k$. The code is constructed by combining  $m-1$  $[n, k]$ vector MDS codes over $\mathbb{F}_q^{\alpha}$, and an $\{(n, k, d), (\alpha, \beta)\}$ functional-repair code from \cite{Wu_regen} for the classical setting.  The code construction in \cite{Wu_regen} is an instance of a functional-repair code for the classical setting, which tolerates an arbitrary number of failures and repairs for the duration of operation of the system.

\vspace{0.1in}

\subsubsection{Lower bound on Intra-Cluster Bandwidth Related Parameters}

We calculate lower bounds on the intra-cluster-bandwidth related parameters $\gamma, \ell', \gamma'$, shown in Fig. \ref{fig:sys_model}, under assumption that \eqref{eq:file_size} is achieved with equality. While studying the impact of any one of these parameters, we ignore the effects of the other two; for example, the lower bound on $\gamma$ is obtained under the assumption that $\ell' = m$ and $\gamma' = \alpha$, etc.

Under functional repair with $d > 0$, the per-node intra-cluster bandwidth needed from the host-cluster is lower bounded by
\begin{eqnarray} \label{eq:gamma}
\gamma & \geq & \gamma^* = \alpha - (d-k+1)^+\beta.
\end{eqnarray}
When $d<k$, the bound gives $\gamma\geq\alpha$, i.e. the entire content of the local helper nodes must be used. For $d \geq k$, at the MBR point characterized by $\alpha = d\beta$, the bound gives $\gamma\geq(k-1)\beta$, and at the MSR point characterized by $\alpha=(d-k+1)\beta$, the bound gives $\gamma \geq 0$. The trivial bound at the MSR point is indeed optimal, since optimal (achieving equality in \eqref{eq:file_size}) codes at the MSR point can be achieved by simply stacking $m$ classical $(n,k,d),(\alpha,\beta)$ MSR codes. In this case, no local help is needed for repair, and hence $\gamma = 0$ is indeed optimal at the MSR point. In fact, under functional repair, the bound in \eqref{eq:gamma} is optimum not just at the MSR point; we prove the converse statement that, as long as $\gamma \geq \gamma^*$, it is indeed possible to achieve the optimal file-size in $\eqref{eq:file_size}$ for any set of parameters, as long as there is a known upper bound on the number of repairs in the system.

We provide bounds for the parameters $\ell'$ and $\gamma'$ which characterize intra-cluster bandwidth from remote helper cluster under the assumption \footnote{The values of $\alpha$ in the  range $(d-k+2)\beta < \alpha \leq (d-k+1)\beta$ corresponds to the region in the trade-off between the minimum-storage operating point and the next corner point. The bound in \eqref{eq:intra_bound_helper2} does not apply when $\alpha$ is in this range.} that $\alpha \geq (d-k+2)\beta$, and $d \geq k$. Specifically we show the parameter $\ell'$ is no less than $m$, i.e., $\ell' = m$, and
\begin{eqnarray}
\gamma' & \geq & \frac{\beta}{m - \ell}. \label{eq:intra_bound_helper2}
\end{eqnarray}
Under functional repair, RLNCs simultaneously optimize usage of both inter-cluster and intra-cluster bandwidths. Our simulations based on RLNCs indicate the tightness (achievability) of the bound in \eqref{eq:intra_bound_helper2}, under functional repair (see Fig. \ref{fig:rlnc_sims}). We do not have an analytical converse for this bound.

The bounds on $\ell'$ and $\gamma'$ highlight the necessary trade-off between the system capacity $B^*$ and the remote helper intra-cluster bandwidth\footnote{We note that we quantify the remote helper intra-cluster bandwidth as $\ell'\gamma'$ though from Fig. \ref{fig:sys_model}, one gets the impression that this is $(\ell'-1)\gamma'$. This discrepancy arises because, in our IFG analysis (see Section \ref{sec:IFG}), we assume the presence of a dedicated external compute node which connects to all the $\ell'$ helper nodes and generates the helper data. Such an assumption makes the IFG modeling symmetric with respect to the helper nodes. The choice of  $\ell'\gamma'$ instead of $(\ell'-1)\gamma'$ as the amount of helper bandwidth is purely matter of convenience. The nature of results do not change even if one assumed $(\ell'-1)\gamma'$ as the amount of helper bandwidth.}  $\ell'\gamma' = m\gamma'$, via parameter $\ell$, the key parameter that distinguishes our model from  the classical model. Our bounds reveal the interesting fact that, while it is beneficial to
increase the number of local helper nodes $\ell$ in order to improve the storage-vs-inter-cluster-repair-bandwidth trade-off, increasing $\ell$ not only increases intra-cluster repair-bandwidth in the host-cluster, but also increases the intra-cluster repair-bandwidth in the remote helper clusters.  For example, if we consider MBR codes (having minimum inter-cluster-repair-bandwidth), we see that their storage-overhead approaches that of MSR codes for large $m$ as $\ell$ gets close to $m$. However, a high value of $\ell$ also increases the remote helper cluster bandwidth; indeed, $m\gamma'$ surges as $m-\ell$ approaches $1$; see Fig. \ref{fig:system_plots} for an illustration.

\vspace{0.1in}

\subsubsection{Security under Passive Eavesdropping - Bounds and Codes}
We study resilience of the clustered storage system against passive eavesdropping. An eavesdropper (say, Eve) gains access to the entire content of any subset of $e$ clusters, where $1 \leq e \leq k$. Eve also gets to observe all the helper data downloaded for repair of nodes in these $e$ clusters. The properties of data collection and disk repair remain same as in the case of no eavesdropper. The setting is along the lines of that considered in \cite{pawar_security}, where authors study security under the classical RC framework. The maximum file size $B^{(s)}$ that can be securely stored such that  Eve does not gain any information about the file is shown to be upper bounded by
\begin{eqnarray}\label{eq:file_size_secure}
B^{(s)} & \leq & \ell (k-e)\alpha + (m - \ell)\sum_{i=e}^{k-1} \min \{\alpha, (d-i)\beta \}.
\end{eqnarray}
We also present explicit optimal secure codes for the MBR point under the setting of exact repair. Like in the case of no security, an optimal secure code is constructed by suitably precoding a combination of $m$ component codes, but this time the component codes themselves are secure codes. A code at the MBR point is constructed by combining $\ell$ secure MDS codes for the wiretap-II channel~\cite{wiretap}\cite{arunkumar_secure}, and $(m-\ell)$ classical exact-repair secure MBR  codes~\cite{pawar_security}\cite{rashmi_secure}. Our security results are straightforward extensions of the security results for classical RCs, given our own results for the no-security case.

\begin{figure*}
	\centering
	\subfigure[Effect of local cluster repair-bandwidth $\gamma$, for an $(n  = 3, k = 2, d = 2), (\ell = 2, m = 3), (\alpha = 8, \beta = 4), \ell' = m, \gamma'=\alpha$ system.]{\label{fig:gamma}\includegraphics[height=1.5in]{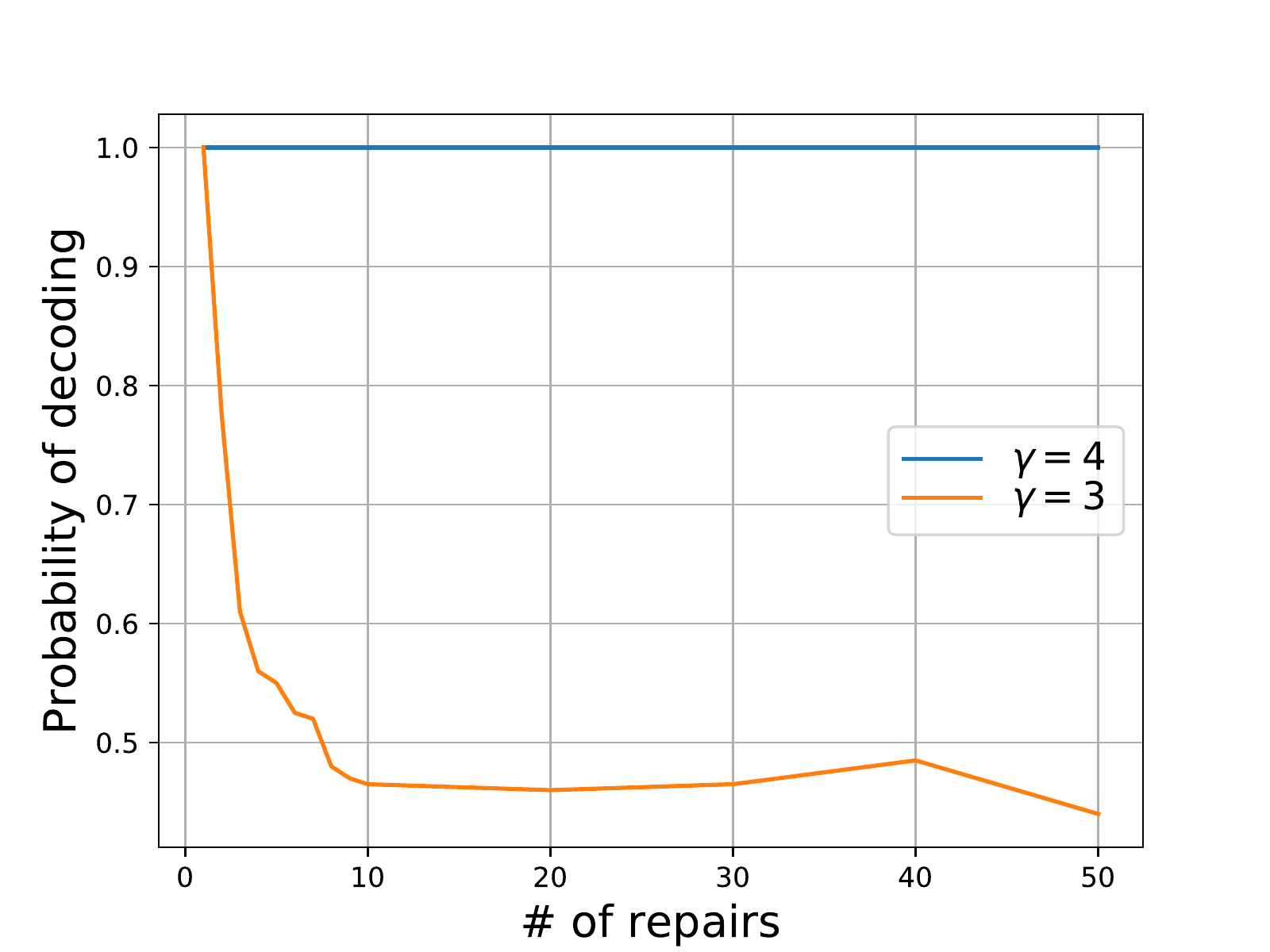}}
	\hspace{0.1in}
	\subfigure[Effect of number of remote helper nodes $\ell'$ for an $(n  = 3, k = 2, d = 2), (\ell = 1, m = 3), (\alpha = 8, \beta = 4), \gamma = \gamma' = \alpha$ system.]{\label{fig:elldash}\includegraphics[height=1.5in]{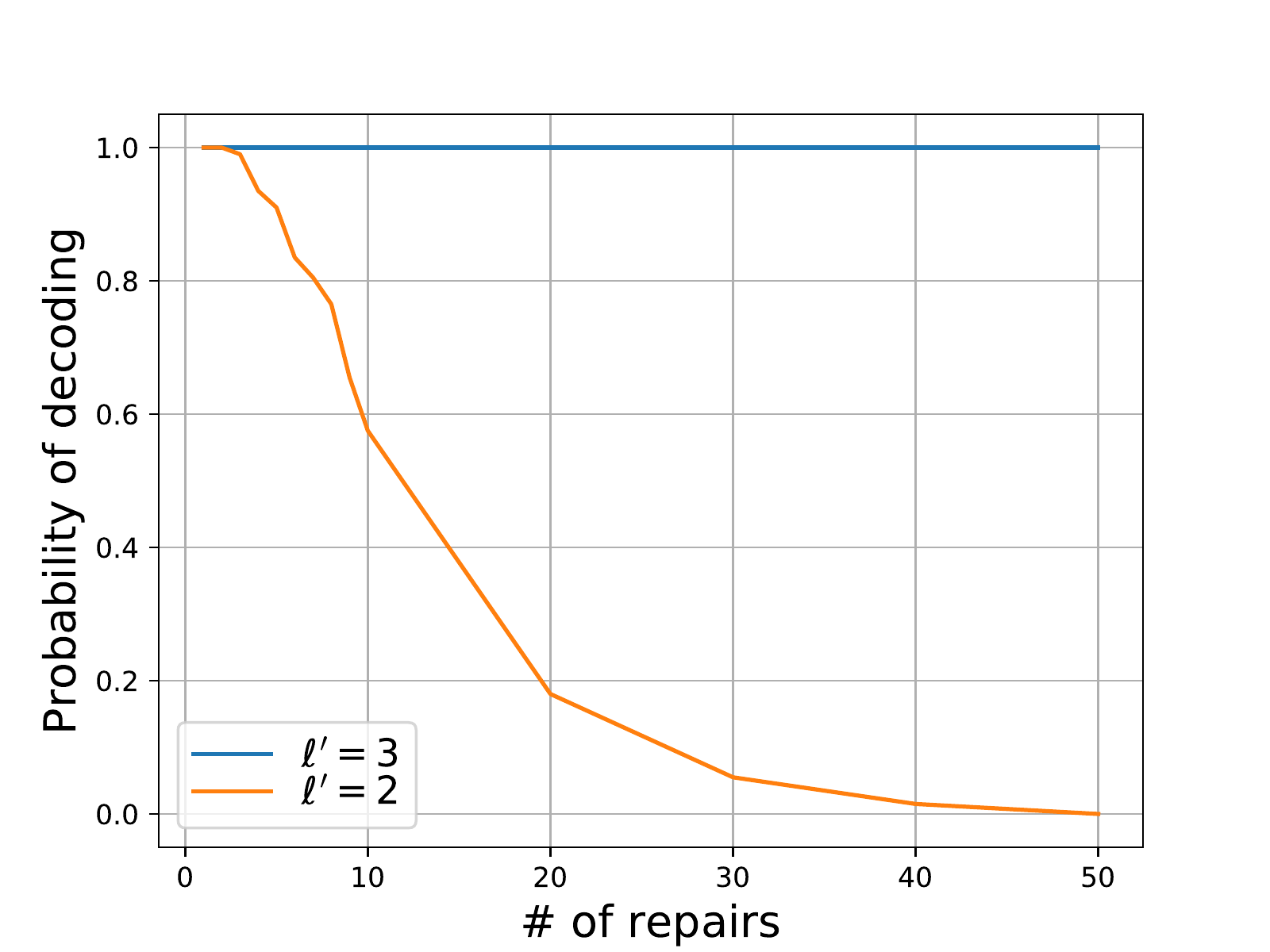}}
	\subfigure[Effect of remote helper cluster repair-bandwidth $\gamma'$, for an $(n  = 3, k = 2, d = 2), (\ell = 1, m = 3), (\alpha = 8, \beta = 4), \ell' = m, \gamma=\alpha$ system.]{\label{fig:gammadash}\includegraphics[height=1.5in]{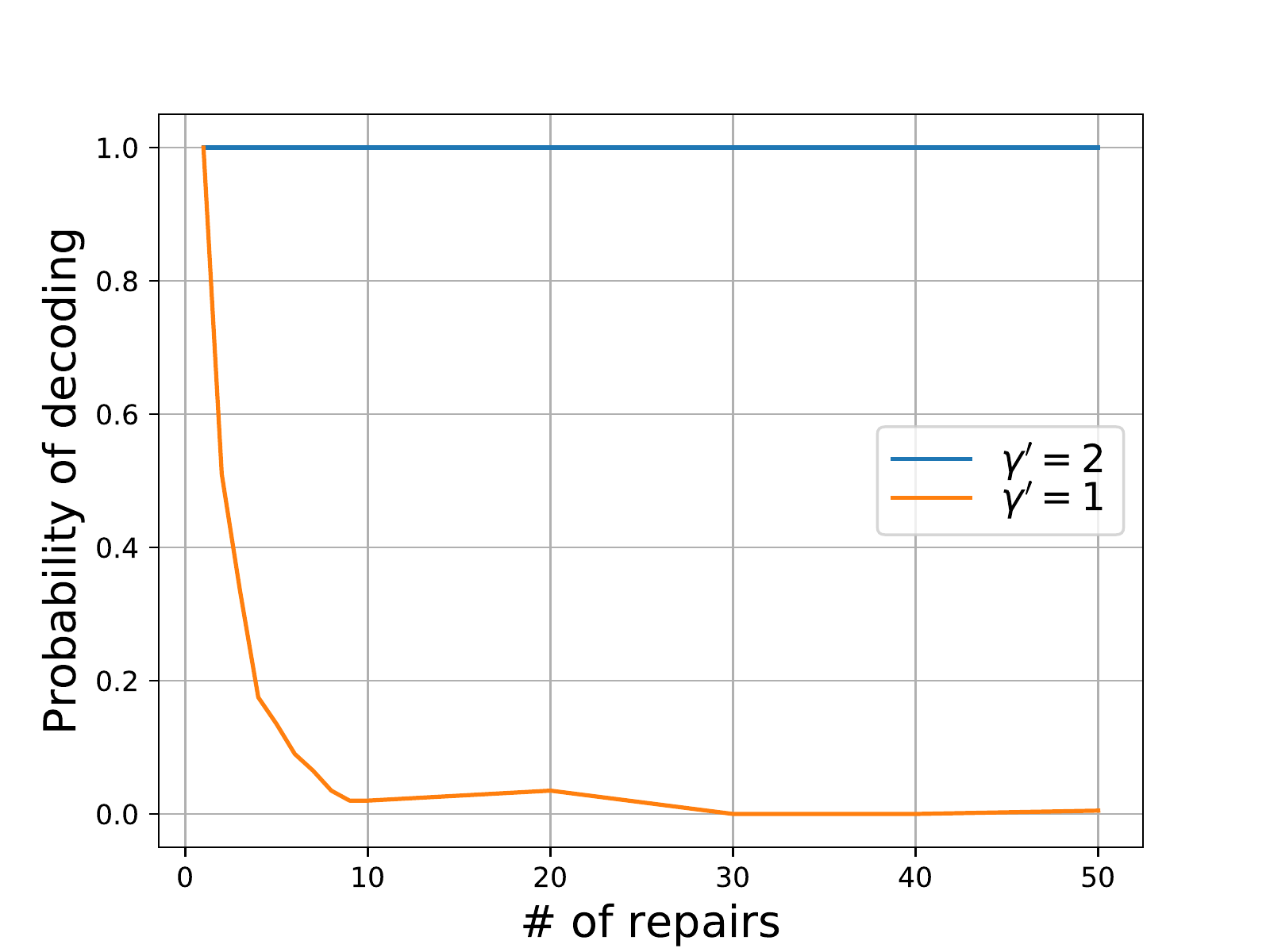}}
	\caption{Simulation results showing probability of successful data collection against number of node-repairs performed, for a clustered storage system employing random linear network codes (RLNCs) with sufficiently large field-size. The three figures respectively indicate the impact of the intra-cluster-bandwidth related parameters $\gamma, \ell'$ and $\gamma'$  on the probability of decoding.}
	\label{fig:rlnc_sims}
\end{figure*}

\begin{figure}
	\centering
	\includegraphics[width=145mm]{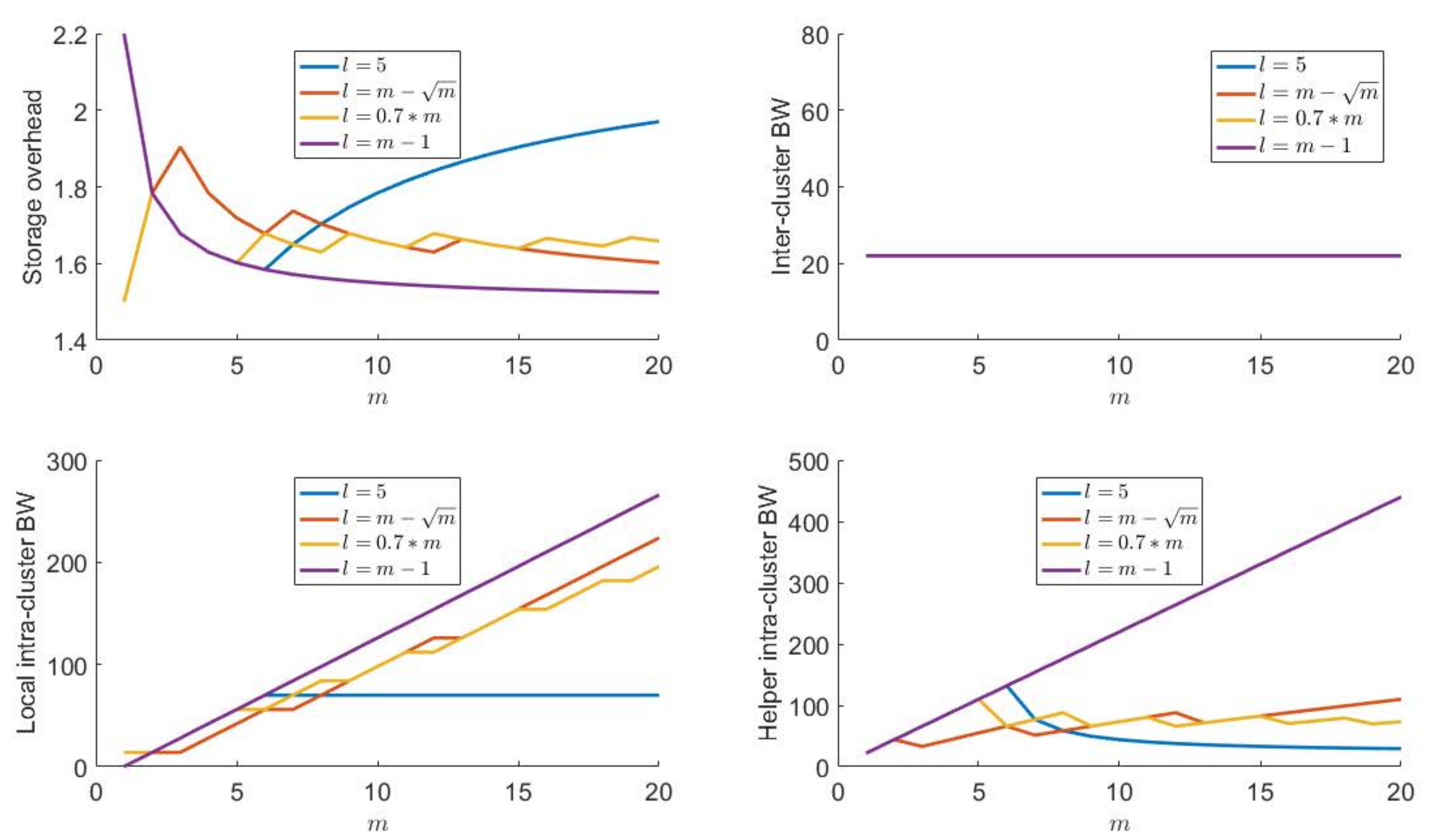}
	\caption{Illustrating the impact of number of local helper nodes on the various performance metrics. We operate at the minimum inter-cluster repair-bandwidth (MBR) point, with parameters $\{(n = 12, k = 8, d = n-1)(\alpha = d \beta, \beta = 2)\}$. Storage-overhead is $\frac{mn\alpha}{B^*}$, where $B^*$ is calculated using \eqref{eq:file_size}. Inter-cluster BW is $d\beta$. Local and helper intra-cluster BWs are respectively calculated using \eqref{eq:gamma} and \eqref{eq:intra_bound_helper2}. We see that while $\ell = m-1$ is ideal in terms of optimizing storage and inter-cluster BW, it imposes the maximum burden on intra-cluster BW.}
	\label{fig:system_plots}
\end{figure}

\subsection{An Example} \label{sec:ex}

\begin{figure}
	\centering
	\includegraphics[width=95mm]{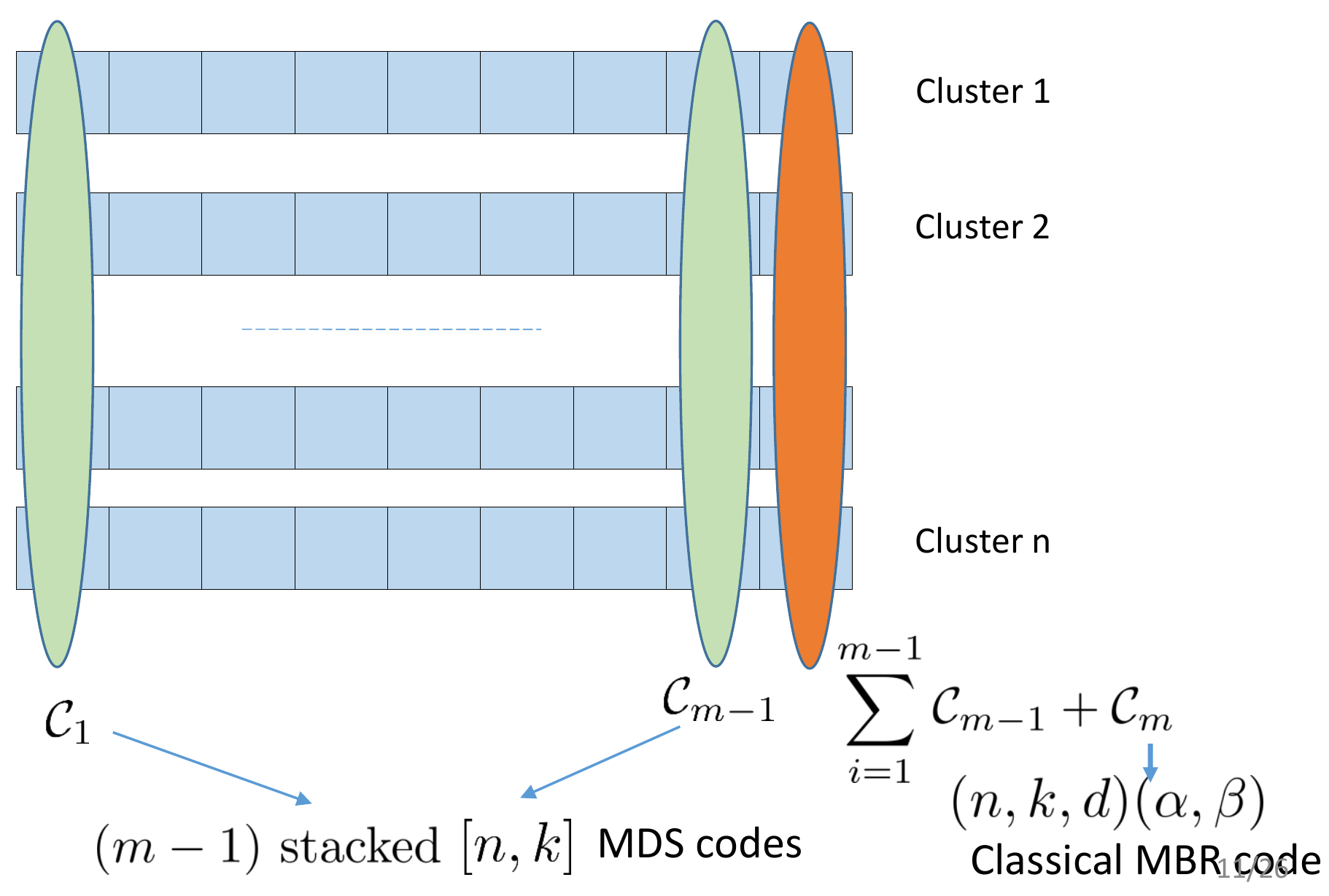}
	\caption{Illustration of an $(n = 4, k = 3, d = 3)(m = 4, \ell = 3)$ generalized regenerating code attaining minimum intra cluster repair-bandwidth.}
	\label{fig:example_code}
\end{figure}

\begin{figure}
	\centering
	\includegraphics[width=85mm]{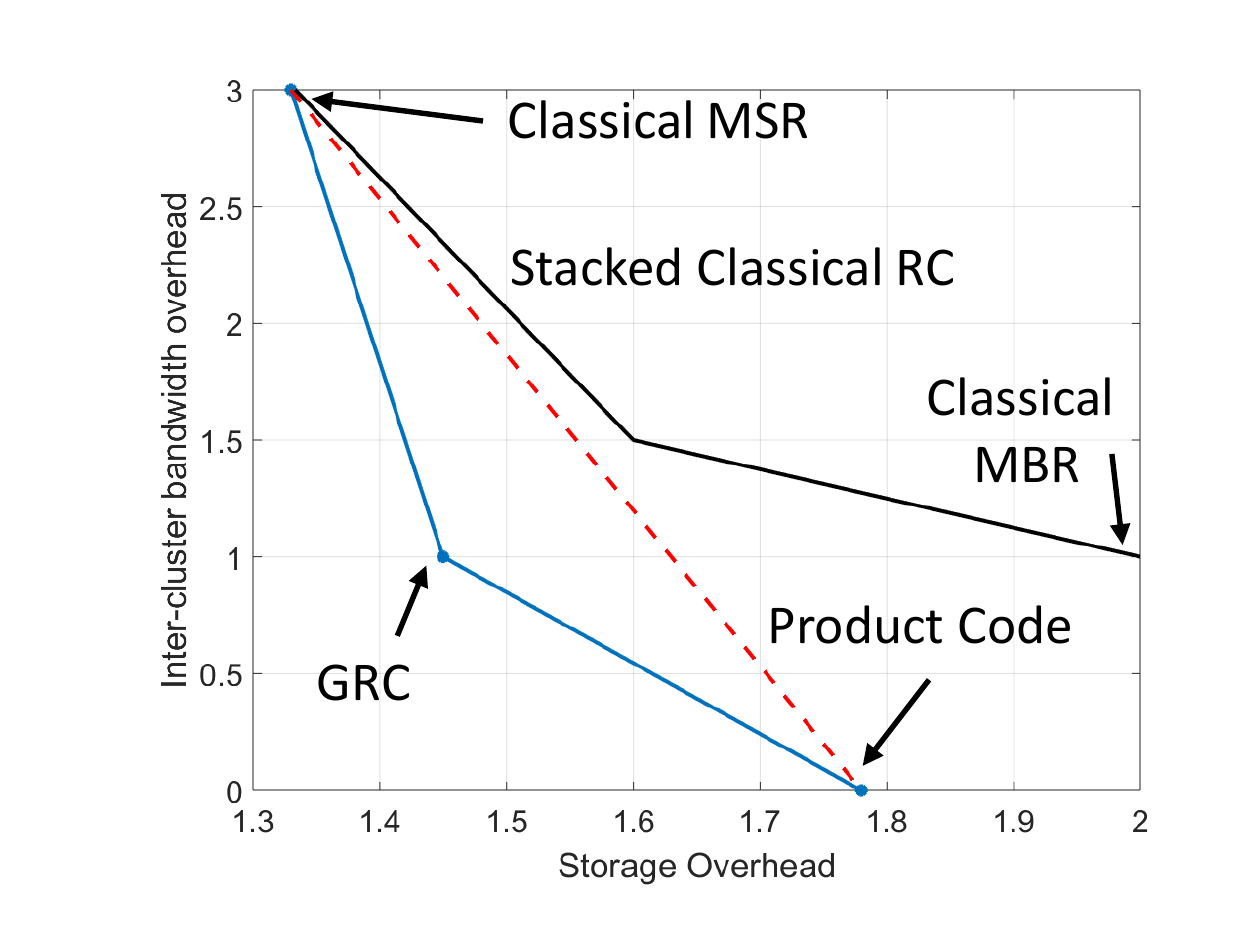}
	\caption{Comparing storage-overhead and inter-cluster repair-bandwidth-overhead of three coding options for an $(n = 4, k = 3)$ clustered storage system, with $m = 4$. The three options are $(i)$ product code comprised of two $[4, 3]$ simple parity check codes, $(ii)$ stacking $4$ $(n = 4, k = 3, d = 3)$ classical regenerating codes, and $(iii)$ an $(n = 4, k = 3, d = 3)$ generalized regenerating code with $\ell = 3$, operating at the MBR point. }
	\label{fig:spaceshare}
\end{figure}

Consider a system consisting of $n = 4$ clusters, with $m = 4$ nodes/cluster, where we have the availability requirement that content from any set of $k = 3$ clusters suffice for data collection. Let us consider three coding options, which permit single node repair: $(i)$ A product code consisting of two $[n = m = 4, k = 3]$  simple parity check codes, one across the clusters, and another within a cluster. In this case, repair happens entirely within a cluster, and there is no inter-cluster bandwidth. Note that this corresponds to the case of $d = 0$ in our framework. $(ii)$ Stacking $m$ classical $(n, k, d = 3)$ RCs - by this we mean that corresponding nodes from the $n$ clusters employ the classical RCs. This corresponds to the case of $\ell = 0$ in our framework. $(iii)$ A $(n, k, d = 3) (m = 4, \ell = 3)$ GRC, which is constructed as follows (see Fig. \ref{fig:example_code}): We stack $m - 1 = 3$  $[n = 4, k = 3]$ simple parity check codes (over the vector alphabet $\mathbb{F}_q^{\alpha}$) to populate coded data in all but the last nodes of all the clusters. In the last column (corresponding to the last node of all clusters), we place the sum of the three previous codes plus a classical $(n = 4, k = 3, d= 3) (\alpha = d\beta, \beta = 1)$ MBR code. Constructions of classical MBR codes appear in \cite{prod_matrix}. Data collection property of the code is straightforward. For repair of any node, the last node of each of the remote helper clusters (with the help of the remaining $3$ nodes) first extracts the MBR code, and passes the helper data for the MBR code to the replacement node. The latter regenerates the MBR code content first, and uses the helper data from the $3$ local nodes to finally recover the original stored content.  The storage-overhead vs inter-cluster-repair-bandwidth trade-off achieved by these three options is shown in Fig. \ref{fig:spaceshare}. We see that the framework of GRCs introduced here, offers operating points which are strictly better than those that can be achieved by space sharing the first two options.

We note that in the above example, the product-code solution carries the extra advantage of handling node-availability during data collection. In other words, when using the product code, one can chose to download content from any $3$ nodes from any of the $k$ clusters. This is sometimes beneficial when the $4^{th}$ node is temporarily unavailable. No such feature is present in the GRC. In the current paper, since we deal with the case of recovery from single node failure, enforcing availability (local dependence within a cluster) essentially implies restricting oneself to using product codes, and nothing else. A key insight that we wish to convey from Fig. 6. is the fact that while dealing with single node repair, it is potentially beneficial to sacrifice availability (of nodes during data collection, offered by product codes) in order to achieve operating points that are strictly better than those obtained by space sharing the two schemes. 

The rest of the document is organized as follows. In Section \ref{sec:special}, we present simple proof of the file-size bound in \eqref{eq:file_size} for the  special cases $d = $ and $\ell =0$. As mentioned before, the general case of the file-size bound is based on the notion of information-flow graph (IFG) models; these IFG models are developed in \ref{sec:IFG}, followed by the derivation of the file-size bound under functional repair for general set of parameters in Section \ref{sec:file_size_bound}.  The exact-repair and functional-repair code constructions appear in Section \ref{sec:code}. Bounds on parameters relating to intra-cluster bandwidth under functional repair, are discussed in Section \ref{sec:intra_cluster_bandwidth}. Section \ref{sec:security} considers security under passive eaves-dropping. Finally, our conclusions and directions for future work appear in Section \ref{sec:con}.

\section{File size bound for Special Cases} \label{sec:special}

In this section, we consider certain special cases of the setting of generalized regenerating codes, and identify the corresponding storage vs inter-cluster-repair-bandwidth trade-offs. The following cases are considered $1)$ $d = 0$, which corresponds to the case when repair of a node is carried out entirely with the help of other local helper nodes in the cluster. $2)$ $\ell = 0$, which corresponds to the case when repair of a node is carried out solely with the help of remote clusters, without taking any help from local nodes in the host-cluster. 

\subsection{Case $d = 0$ : No Inter-Cluster Help for Repair}

When $d = 0$, node repair is accomplished by contacting any set of $\ell$ other nodes in the host cluster. Since any set of $k$ clusters should be sufficient to decode the whole file, it follows that the file-size $B$ is upper bounded by
\begin{eqnarray}
	B & \leq & \ell k \alpha.
\end{eqnarray}
The achievability of the above bound follows by using an $[n, k] \times [m, \ell]$ product code, with both component codes being MDS codes over $\mathbb{F}_q^{\alpha}$. In fact, the parameter $\alpha$ is redundant for this case; one may choose $\alpha = 1$ in practice. It is clear that there is no storage vs inter-cluster-repair-bandwidth trade-off offered by this case.

\subsection{Case $\ell = 0$: No local helper data} \label{sec:ell0}

We next consider the special case when repair is performed without any help from nodes in the host-cluster. Let $\mathcal{C}_m$ denote a GRC of file-size $B$ with parameters $\{(n, k, d)(\alpha, \beta)(m, \ell = 0)\}$. Let $B^*_m$ denote that maximum possible file-size for any GRC having parameters $\{(n, k, d)(\alpha, \beta)(m, \ell = 0)\}$.  We note that the code $\mathcal{C}_1$ denotes a classical RC, and under functional-repair, we know that~\cite{dimakis}
\begin{eqnarray} \label{eq:dimakis_opt}
	B_1^* & = & \sum_{i=0}^{k-1} \min \{\alpha, (d-i)\beta \}.
\end{eqnarray}

\begin{thm} \label{thm:no_side_info}
	The optimal file size under the setting of functional-repair GRCs, for the case of no local helper nodes, is given by
	\begin{eqnarray}
		B^*_m & = & mB_1^* \ = \ m\sum_{i=0}^{k-1} \min \{\alpha, (d-i)\beta \}.
	\end{eqnarray}
\end{thm}
\begin{proof}
	The achievability part of the proof is straightforward; the optimal code is constructed by simply stacking $m$ classical codes $\mathcal{C}_1$ each of which achieves the bound in \eqref{eq:dimakis_opt}. By stacking, we mean that the code $\mathcal{C}_1$ is deployed across the corresponding nodes from all $n$ clusters. In this case, note that during node-repair, there is no pooling of content from various nodes of a remote helper cluster; repair happens as though there is only one code $\mathcal{C}_1$ in the system.
	
	For showing the upper bound on the file size, we note that given a code $\mathcal{C}_m$ with file-size $B$ having parameters $\{(n, k, d)(\alpha, \beta)(m, \ell = 0)\}$, one can construct a functional-repair classical regenerating code $\widehat{\mathcal{C}}_1$, also with file-size $B$, and having parameters $\{(n, k, d), (m\alpha, \widehat{\beta})\}$, where $\widehat{\beta} \leq m\beta$. For this, we simply assume the contents of all $m$ nodes of any cluster $i$ of $\mathcal{C}_m$, to be the contents of node $i$ of  $\widehat{\mathcal{C}}_1, 1 \leq i \leq n$. Clearly $\widehat{\mathcal{C}}_1$ retains the data collection property. For node repair in $\widehat{\mathcal{C}}_1$, we perform individual repairs of each of the $m$ nodes, but with the same set of $d$ remote helper clusters. In this case, we know that
	\begin{eqnarray}
		B & \leq & \sum_{i=0}^{k-1} \min \{m\alpha, (d-i)\widehat{\beta} \} \\
		& \leq & \sum_{i=0}^{k-1} \min \{m\alpha, (d-i)m{\beta} \} \\
		& = & \ m\sum_{i=0}^{k-1} \min \{\alpha, (d-i)\beta \}.
	\end{eqnarray}
\end{proof}

\vspace{0.1in}

Theorem \ref{thm:no_side_info} implies that, under functional repair, the normalized trade-off between storage-overhead $nm\alpha/B$ and inter-cluster-repair-bandwidth-overhead $d\beta/\alpha$ for an $\{(n, k, d)(\alpha, \beta)(m, \ell = 0)\}$ GRC is identical for any $m$; specifically, it is identical to the trade-off of an $\{(n, k, d)(\alpha, \beta)\}$ functional-repair classical RC.

\section{Information Flow Graph Model} \label{sec:IFG}

In this section, we describe the information flow graph (IFG) models used to derive the various bounds in this work. The models are generalizations of the one used in \cite{dimakis} for the case of classical regenerating codes. Under functional repair, the problem is one of multicasting the source file to an arbitrary number of data collectors over the IFG. The IFG characterizes the data flows from the source to a data collector, and also reflects the sequence of failures and repairs in the storage system. Two models of IFGs will be used; the first one will be used in two  scenarios: $1)$ to derive the trade-off between storage-overhead and inter-cluster repair-bandwidth overhead. While obtaining this trade-off, we ignore the effects of intra-cluster bandwidth, $2)$ to find the optimal local helper node intra-cluster bandwidth $\gamma$, which is needed to establish the optimal trade-off between storage-overhead and inter-cluster repair-bandwidth overhead. We wish to note that while obtaining the bound on $\gamma$, we do not impose any limitations on $\gamma', \ell'$, i.e, we assume that $\gamma' = \alpha$ and $\ell' = m$. A second related  model will be used while deriving the lower bounds on the parameters $\ell', \gamma'$, which relate to the intra-cluster repair bandwidth needed in the remote helper clusters. In this second model, we shall assume that $\gamma = \alpha$, i.e., we ignore the effects of limited local helper-node intra-cluster bandwidth, while calculating bounds on remote helper-node intra-cluster bandwidth. We describe the two models next.

\subsection{IFG Model for Storage vs Inter-Cluster Repair Bandwidth Trade-off} \label{sec:IFG_inter}

\begin{figure*}
	\centering
	\includegraphics[height=2.5in]{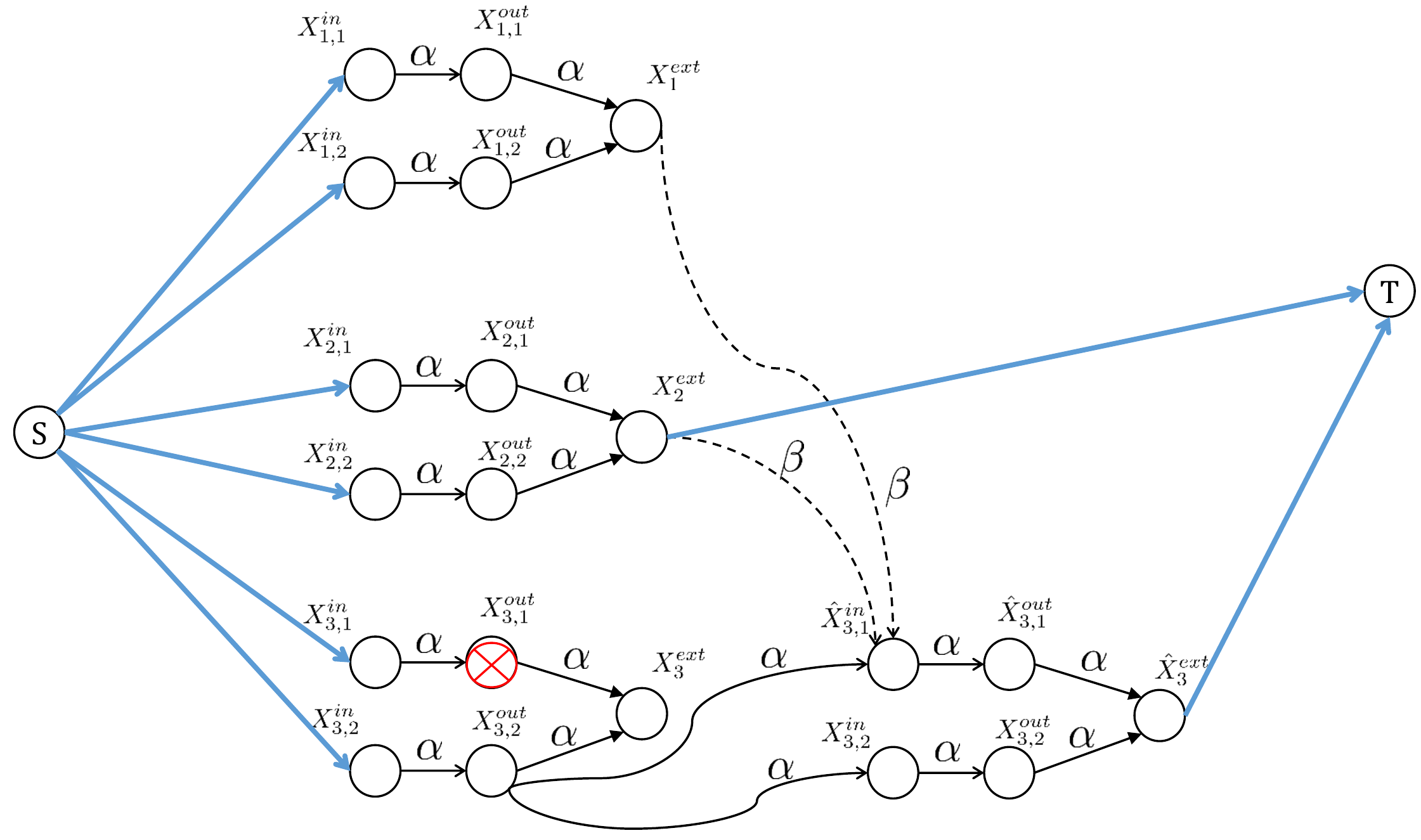}
	\caption{An example of the IFG model representing the notion of generalized regenerating codes, when intra-cluster bandwidth is ignored. In this figure, we assume $(n = 3, k = 2, d= 2) (m = 2, \ell = 1)$. The model is a generalization of one the used in \cite{dimakis} for the setting of classical regenerating codes.}
	\label{fig:IFG_inter}
\end{figure*}

Let $X_{i, j}$ denote the physical node $j$ in cluster $i$, $ 1 \leq i \leq n, 1 \leq j \leq m$. Recall that capacity of any node is $\alpha$. In the IFG, the physical node is represented by the pair of nodes $X_{i,j}^{in}$ and $X_{i,j}^{out}$, with an edge of capacity $\alpha$ going from $X_{i,j}^{in}$ to $X_{i,j}^{out}$. The nodes $X_{i,j}^{in}$ and $X_{i,j}^{out}$ will be respectively referred to as the in-node and out-node corresponding to the physical node $X_{i,j}$.  We will write $(X_{i,j}^{in} \to X_{i,j}^{out})$ to denote that there is an edge from going from $X_{i,j}^{in}$ to $X_{i,j}^{out}$.  With a slight abuse of notation, we will let $X_{i,j}$ to also denote the pair $(X_{i,j}^{in}, X_{i,j}^{out})$ of the graph nodes. Cluster $i$  also has an additional external node, denoted as $X_i^{ext}$. Each out-node $X_{i, j}^{out}, 1  \leq j \leq m$ is connected to $X_i^{ext}$ via an edge of capacity $\alpha$. The external node $X_i^{ext}$ is used to transfer data outside the cluster, and thus serves two purposes: $1)$ it represents a single point of contact to the cluster, for a data collector which connects to this cluster, and $2)$ it represents the compute unit which generates the $\beta$ symbols for repair of any node in a different cluster.

The source node $S$ represents the original placement of the encoded source file into the $nm$ storage nodes. $S$ connects to the in-nodes of all physical storage nodes in their original state $(S\to X_{i,j}^{in}), \forall i\in [n], \forall j\in [m]$, via links of infinite capacity. The sink node $T$ represents a data collector, it connects to the external nodes of an arbitrary subset of $k$ clusters $(X_{i}^{ext}\to T)$ also via links of infinite capacity.

Each cluster at any moment has $m$ \emph{active} nodes. When a physical node $X_{i,j}$ fails, it becomes \emph{inactive}, and its replacement node, say  $\widehat{X}_{i,j}$, becomes active instead (see Fig. \ref{fig:IFG_inter} for an illustration). The replacement node $\widehat{X}_{i,j}$ is regenerated by downloading $\beta$ symbols from any $d$ nodes in the set $\{X_{i'}^{ext}, 1 \leq i' \leq n, i' \neq i\}$. The replacement node also connects to any subset of $\ell$ nodes in the set $\{X_{i, j'}^{out}, 1 \leq j' \leq m, j' \neq j \}$. The capacity of the links $\{(X_{i, j'}^{out} \to \widehat{X}_{i,j}^{in}), 1 \leq j' \leq m, j' \neq j \}$ depend on whether we use the model for finding the inter-cluster-bandwidth vs storage trade-off, or we use it for finding bounds on local helper bandwidth $\gamma$. These links have capacity $\alpha$ and $\gamma$ in the former and latter cases, respectively.

In our model, recall that we focus on one repair at a time. In this scenario, along with the replacement of $X_{i,j}$ with $\widehat{X}_{i,j}$, we will also \textit{copy} all the remaining $m-1$ nodes, as they are, in the cluster $i$, and represent them with new identical pair of nodes  $(X_{i,j'}^{in}, X_{i,j'}^{out}), 1 \leq j' \leq m, j' \neq j$. We shall also a have a new external node for the cluster, which connects to the new $m$ out-nodes. Thus, in the IFG modeling, we say that the entire old cluster (where the failed node resides) becomes \emph{inactive}, and gets replaced by a new \emph{active} cluster. For either data collection or repair, we connect to external nodes of the active clusters. Note that, at any point in time, a physical cluster contains only one active cluster in the IFG, and $f_i$ inactive clusters in the IFG, where $f_i \geq 0$ denotes the total number of failures and repairs experienced by the various nodes in the cluster. We shall use the notation $\mathcal{X}_i(t), 0 \leq t \leq f_i$ to denote the cluster that appears in IFG after the $t^{\text{th}}$ repair associated with cluster $i$. The clusters  $\mathcal{X}_i(0), \ldots,  \mathcal{X}_i(f_i-1)$ are inactive, while  $\mathcal{X}_i(f_i)$ is active, after $f_i$ repairs.
The nodes of  $\mathcal{X}_i(t)$ will be denoted by $X_{i,j}^{in}(t), X_{i,j}^{out}(t), {X}_{i}^{ext}(t), 1 \leq j \leq m$. With a slight abuse of notation, we will let $\mathcal{X}_i(t)$ to also denote the collection of all $2m + 1$ nodes in this cluster.  We write $X_{i, j}(t)$ to denote the pair $(X_{i,j}^{in}(t), X_{i,j}^{out}(t))$; again, with a slight abuse of notation,  we shall use $X_{i, j}(t)$ to also denote the node $j$ in  cluster $i$ after the $t^{\text{th}}$ repair (in cluster $i$). We further use notation $Fam(i)$ to denote the union of all nodes in all inactive clusters, and the active cluster, corresponding to the physical cluster $i$ after $t$ repairs in cluster $i$, i.e., $Fam(i) = \cup_{t = 0}^{f_i}\mathcal{X}_i(t)$. We have avoided indexing $Fam(i)$ with the parameter $t$ as well, to keep the notation simple. The value of $t$ in our usage of the notation
$Fam(i)$ will be clear from the context.

\begin{figure*}
	\centering
	\includegraphics[height=2.5in]{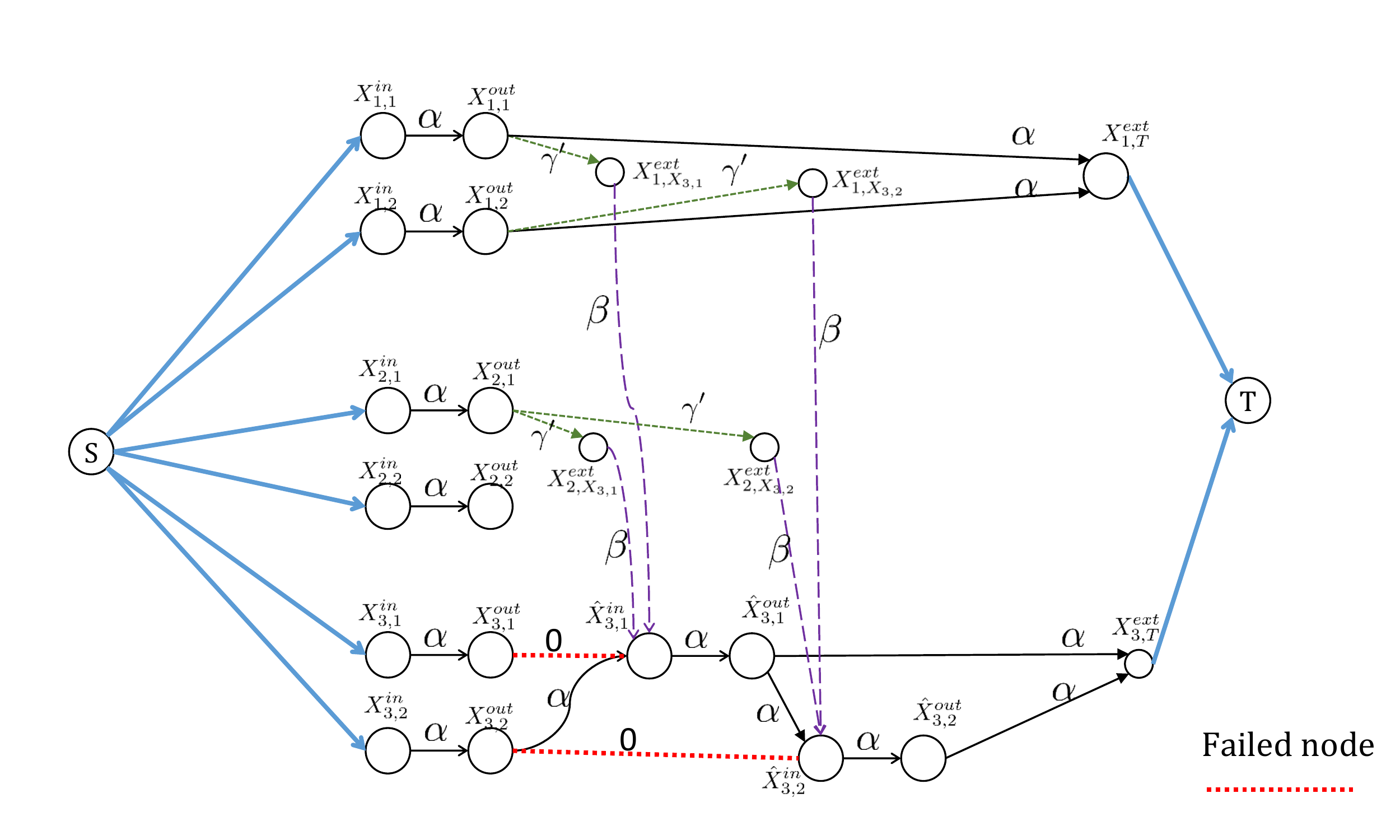}
	\caption{An example of the second IFG model, with limited intra-cluster bandwidth in the remote helper clusters. The model assumes that $\gamma = \alpha$. In this example, we assume $(n = 3, k = 2, d= 2) (m = 2, \ell = 1) (\ell'=1)$.}
	\label{fig:IFG_inter_2}
\end{figure*}

\subsection{IFG Model for Finding Bounds on Intra-Cluster Repair Bandwidth} \label{sec:IFG_intra}
We now describe the model used to obtain lower bounds on the parameters $\gamma', \ell'$. Unlike in the case of bounds for file-size $B$  and local helper node intra-cluster bandwidth $\gamma$, where we also show converses, for $\gamma'$ and $\ell'$ we do not provide converses to the lower bounds. When only dealing with lower bounds, we can significantly simplify the model described above. In addition to making these simplifications, we also add some structure to the model to enable usage of $\gamma'$ and $\ell'$. We describe these changes next.

In the second model, each physical node $X_{i,j}$ is again represented by the pair of nodes $(X_{i,j}^{in}, X_{i,j}^{out})$, such that the edge $(X_{i,j}^{in} \to X_{i,j}^{out})$ has capacity $\alpha$. In this model, an external node(s) is(are) added dynamically to a cluster whenever it aids in either data collection or repair of another node. Each instance of an external node is used exactly once, either for a repair or a data collection operation.  Whenever a physical node $X_{i,j}$ fails, we say that it becomes \emph{inactive}, and its replacement node, say  $\widehat{X}_{i,j} = (\widehat{X}_{i,j}^{in}, \widehat{X}_{i,j}^{out})$, becomes \emph{active} in the same cluster. The remaining $m-1$ nodes are not replicated, as in the previous model.
Thus, in the second model, there is only a single graph cluster corresponding to a physical cluster. On the graph representation, the replacement node is visually linked to the replaced failed node with a red dotted line (see Fig. \ref{fig:IFG_inter_2} for an example).

We explain the repair and data collection operation in the IFG in more detail next. During repair, the replacement node $\widehat{X}_{i,j}$ connects to any subset of $\ell$ active nodes in the same cluster via links of capacity $\alpha$. It also downloads $\beta$ symbols each from $d$ remote helper clusters via their external nodes. If replacement node $\widehat{X}_{i,j}$ downloads helper data from cluster $i'$, an external node, $X_{i',X_{i,j}}^{ext}$, is added to the IFG, such that the edge $(X_{i',X_{i,j}}^{ext} \to \widehat{X}_{i,j}^{in})$ has capacity $\beta$.  External node $X_{i',X_{i,j}}^{ext}$ also connects locally to a subset of $\ell'$ active out-nodes of cluster $i'$ via links of capacity $\gamma'$. Note how we index the external node of cluster $i'$ that aids in the repair of ${X}_{i,j}$. Every time cluster $i'$ acts as a remote helper cluster toward the repair of any node, we add a new external node in a manner similar to $X_{i',X_{i,j}}^{ext}$.
Finally, a data collector, $T$, connects to cluster $i$ via the external node $X_{i, T}^{ext}$, which  in turn connects to all $m$ active out-nodes in the cluster via links of capacity $\alpha$. We index the external node also with $T$ since the $m$ active nodes that form part of the cluster evolves over time.

In comparison with previous model, we do not time-index the sequence of failures in the current model. This is because, in our proof of bounds for $\gamma'$ and $\ell'$, we only consider system evolutions in which each node fails at most once. In this case, we find it convenient simply to  denote the replacement node of $X_{i, j}$ as $\widehat{X}_{i,j}$.

\section{File size bound for General Parameters} 
\label{sec:file_size_bound}

In this section, we derive the file-size bound in \eqref{eq:file_size} under the setting of functional repair, for arbitrary set of code parameters. We further use this bound to characterize the storage-overhead vs inter-cluster-repair-bandwidth-overhead trade-off. Intra-cluster bandwidth is ignored in this section. Thus, for the repair of any node, the entire content of $\ell$ local helper nodes can be used; similarly,  the entire content ($m\alpha$ symbols) of each remote helper cluster is used to generate its $\beta$ helper symbols.

\vspace{0.1in}

\begin{thm} \label{thm:file_size}
	The file size $B$ of a functional repair generalized regenerating code having parameters $\{(n, k, d)$ $(\alpha, \beta)$ $(m, \ell)\}$ is upper bounded by
	\begin{eqnarray}
		B & \leq & B^*  \ = \ \ell k\alpha + (m - \ell)\sum_{i=0}^{k-1} \min \{\alpha, (d-i)^+\beta \}. \label{eq:file_size_1}
	\end{eqnarray}
	Further, if there is a known upper bound on the number of repairs that occur for the duration of operation of the system, the above bound is sharp, i.e., $B^*$ gives the functional repair storage capacity of the  system.
\end{thm}
\begin{proof}
	The proof technique is similar to the proof of the bound under functional repair for the setting of classical regenerating codes~\cite{dimakis}. The problem of functional repair is one of multicasting, and thus for finding the desired upper bound on the file-size, it is enough if we exhibit a cut in an IFG, for a specific sequence of failures and repairs, which separates the source from the sink, such that the value of the cut is the desired upper bound. We shall then show that, for any valid IFG, independent of the specific sequence of failures and repairs,  $B^*$ is indeed a lower bound on the minimum possible value of any $S-T$ cut. The achievability result, when there is known upper bound on the number of failures and repairs, will then follow from results in network coding~~\cite{KoetterMedard}.
	
	We begin with the proof of the upper bound. We consider a sequence of $k(m-\ell)$ failures and repairs, as follows: Physical nodes $X_{i, \ell + 1}, X_{i,\ell+2}, \ldots X_{i,m}$ fail in this order in cluster $i = 1$, then in cluster $i=2$, and so on, until cluster $i=k$. In the IFG, (see Section \ref{sec:IFG_inter}), this corresponds to the sequence of failures of nodes $X_{1, \ell+1}(0), X_{1, \ell+2}(1), \ldots, X_{1, m}(m - \ell-1), X_{2, \ell+1}(0), \ldots, X_{2, m}(m - \ell-1), \ldots, X_{k, m}(m - \ell-1)$, in the respective order. The replacement node $X_{i,\ell + t}(t)$ for $X_{i,\ell+t}(t-1), 1 \leq t \leq m-\ell$ draws local helper data from $X_{i, 1}(t-1), X_{i,2}(t-1), \cdots, X_{i,\ell}(t-1)$, and  remote helper data from the clusters $\mathcal{X}_1(m-\ell),\cdots,\mathcal{X}_{i-1}(m-\ell)$ and from some set of $d-\min\{i-1, d\} = (d-i+1)^+$ other active clusters in the IFG. An example is shown in Fig. \ref{fig:cut} for a set of system parameters that is same as those used in Fig. \ref{fig:IFG_inter}.

	\begin{figure}
		\centering
		\includegraphics[width=115mm]{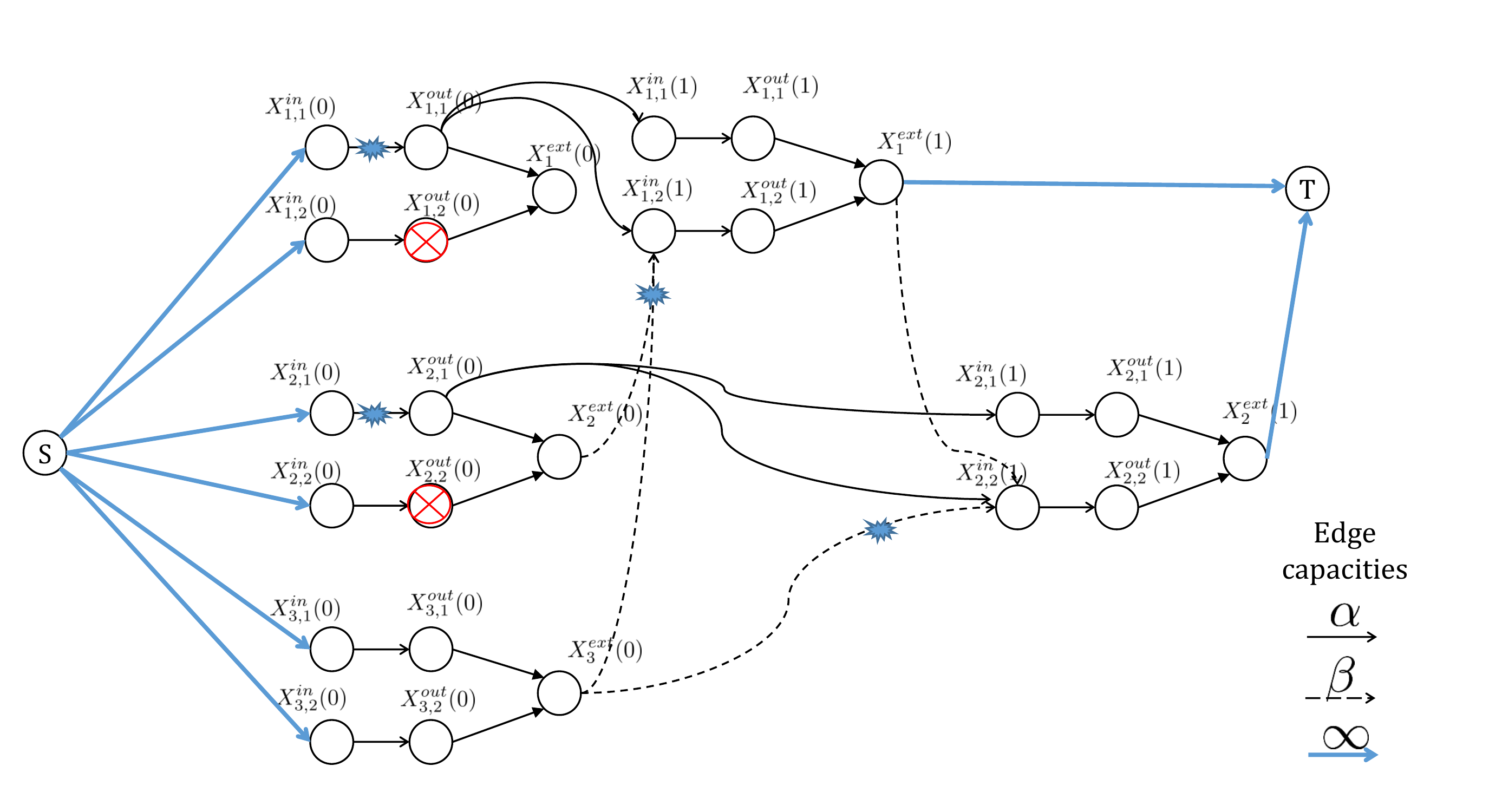}
		\caption{An example of the information flow graph used in cut-set based upper bound for the file-size. In this figure, we assume $(n = 3, k = 2, d= 2) (m = 2, \ell = 1)$. We also indicate a possible choice of $S-T$ cut that results in the desired upper bound.}
		\label{fig:cut}
	\end{figure}
	
	Let data collector $T$ connect to clusters $\mathcal{X}_1(m-\ell), \ldots, \mathcal{X}_k(m-\ell)$. Consider the $S-T$ cut consisting of the following edges of the IFG:
	\begin{itemize}
		\item $\{(X_{i,j}^{in}(0) \to X_{i,j}^{out}(0), i\in [k], j\in [\ell]\}$.  Total capacity of these edges is $kl\alpha$.
		\item For each $i\in [k], t \in [m -\ell]$, either the set of edges $\{(X_{i'}^{ext}(0) \to X_{i,\ell+ t}^{in}(t)), i'\in \{$remote helper cluster indices for the replacement node $X_{i,\ell+ t}^{in}(t) \} \backslash [\min\{i-1, d\}]\}$,  or the edge $(X_{i,\ell + t}^{in}(t) \to  X_{i,\ell+t}^{out}(t))$. Between the two possibilities, we pick the one which has smaller capacity. In this case, the total capacity of this part of the cut is given by $\sum_{i=1}^{k} \sum_{j=\ell+1}^{m} \min \{\alpha, (d-\min\{i-1, d\})\beta \} = (m - \ell)\sum_{i=1}^{k} \min \{\alpha, (d-i+1)^+\beta \}$.
	\end{itemize}
	The value of the cut is given by $kl\alpha +  (m - \ell)\sum_{i=1}^{k} \min \{\alpha, (d-i+1)^+\beta \} = B^*$, which proves our upper bound. In the example on Fig. \ref{fig:cut} for $(n = 3, k = 2, d= 2) (m = 2, \ell = 1)$, first, $m-\ell=1$ node fails in cluster $1$ and downloads helper data from clusters $2,3$, second, a node fails in cluster $2$ and downloads helper data from clusters $1,3$. The data collector connects to clusters $1,2$. A minimal cut for $2\beta\leq \alpha$ is shown on the figure, and has value $2\alpha + 3\beta=B^*$.
	
	We next show that for any valid IFG (independent of the specific sequence of failures and repairs),  $B^*$ is indeed a lower bound on the minimum possible value of any $S-T$ cut. Consider any $S-T$ cut, and let IFG$_S$ and IFG$_T$ denote the two resultant disconnected parts of the IFG corresponding to the nodes $S$ and $T$, respectively.  Since node $T$ connects to $k$ external nodes via links of infinite capacity, we only consider cuts such that  IFG$_T$ has at least $k$ external nodes corresponding to active clusters. Next, we observe that the IFG is a directed acyclic graph, and hence there exists a topological sorting of nodes of the graph such that  an edge exists between two nodes $A$ and $B$ of the IFG only if $A$ appears before $B$ in the sorting. Further, we consider a topological sorting such that
	all in-, out- and external nodes of the cluster $\mathcal{X}_i(\tau)$ appear together in the sorted order, $\forall i,\tau$.

	Now, consider the sequence $\mathcal{E}$ of all the external nodes (which are part of both active and inactive clusters) in IFG$_T$ in their sorted order. Let $Y_1$ denote the first node in this sequence. Without loss of generality let $Y_1 \in Fam(1)$. Next, consider the subsequence of $\mathcal{E}$ which is obtained after excluding all the external nodes in $Fam(1)$ from $\mathcal{E}$. Let $Y_2$ denote the first  external node in this subsequence.  We continue in this manner until we find the first $k$ external nodes $\{ Y_1, Y_2, \ldots, Y_k\}$ in $\mathcal{E}$, such that each of the $k$ nodes corresponds to a distinct physical cluster. Once again, without loss of generality, we assume that $Y_i \in Fam(i), 2 \leq i \leq k$. Let us assume that $Y_i = X_i^{ext}(t_i)$, for some $t_i$. Now, consider the $m$ out-nodes $X_{i, 1}^{out}(t_i), \ldots, X_{i, m}^{out}(t_i)$ that connect to $X_i^{ext}(t_i)$. Among these $m$ out-nodes, let $a_i, 0 \leq a_i \leq m$ denote the number of out-nodes that appear in IFG$_S$.  Without loss of generality let these be the nodes $X_{i, 1}^{out}(t_i), X_{i, 2}^{out}(t_i), \ldots, X_{i, a_i}^{out}(t_i)$. Next, corresponding to the out-node $X_{i, j}^{out}(t_i), a_i + 1 \leq j \leq m$, consider its past versions $\{ X_{i, j}^{out}(t), t < t_i\}$ in the IFG, and let $X_{i, j}^{out}(t_{i,j})$, for some $t_{i, j} \leq t_i$ denote the first sorted node that appears\footnote{It may be noted that even though $X_{i, j}^{out}(t_{i,j})$ appears in IFG$_T$, the corresponding external node $X_{i}^{ext}(t_{i,j})$ appears in IFG$_S$. This is due to our assumption that $Y_i = X_i^{ext}(t_i)$ is the first external node, corresponding to physical cluster $i$, that appears in IFG$_T$. } in  IFG$_T$. Without loss of generality, let us also assume that the nodes  $\{X_{i, j}^{out}(t_{i,j}), a_i + 1 \leq j \leq m\}$ are sorted in the order $X_{i, a_i+1}^{out}(t_{i,a_i + 1}),  X_{i, a_i+2}^{out}(t_{i,a_i + 2}), \ldots, X_{i, m}^{out}(t_{i,m})$. An illustration is provided in Fig. \ref{fig:cut_converse}.
	
	\begin{figure}
		\centering
		\includegraphics[width=115mm]{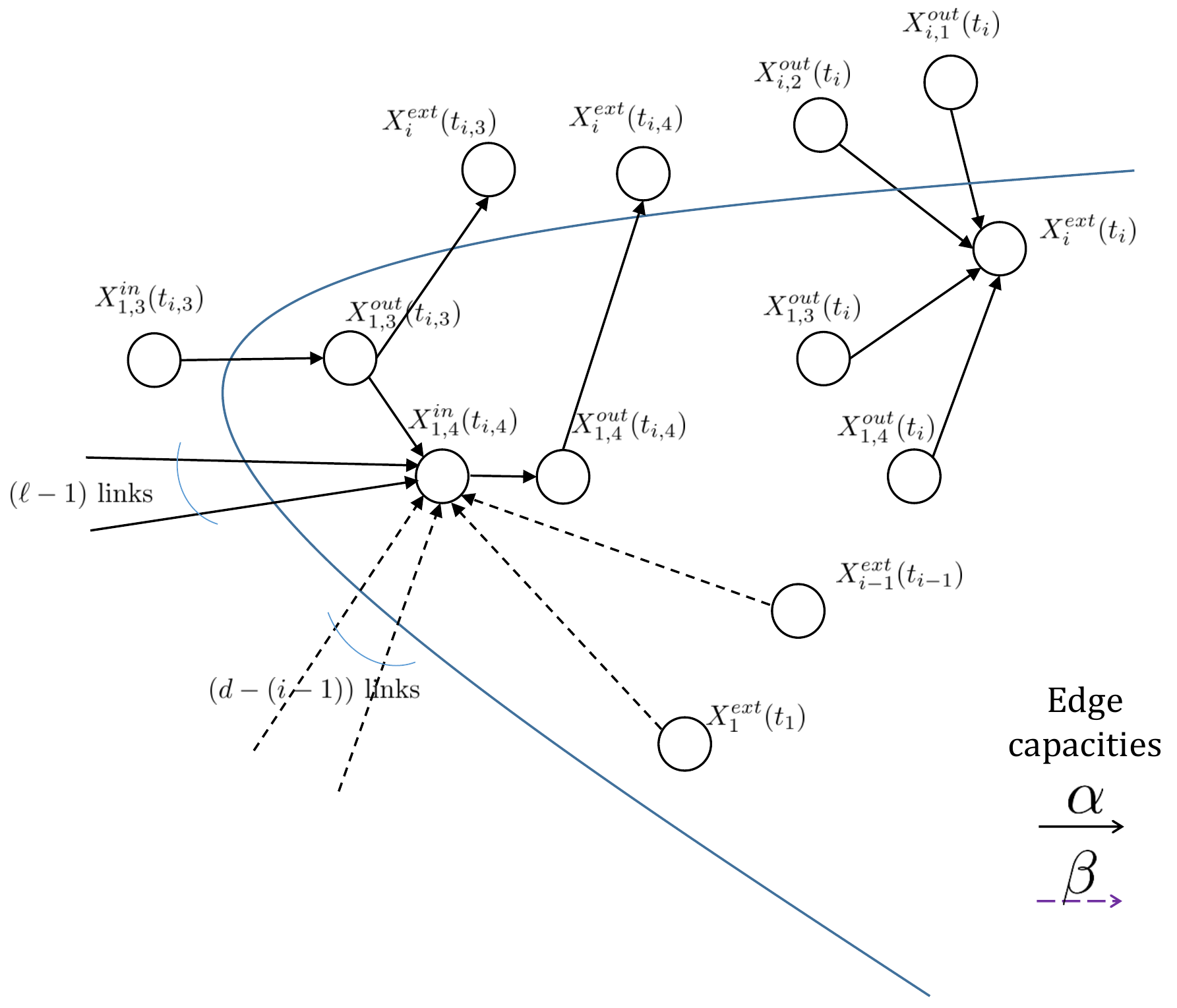}
		\caption{An example of how any $S-T$ cut in the IFG affects nodes in $Fam(i)$.
			In the example, we assume $m = 4$. With respect to the description in the text, $a_i = 2$. Further, 	the node $X_{i, 4}(t_{i, 4})$ is a replacement node in the IFG. }
		\label{fig:cut_converse}
	\end{figure}
	
	To obtain a lower bound on the value of the $S-T$ cut, we make the following observations:
	\begin{itemize}
		\item The $a_i$ edges $\{(X_{i, j}^{out}(t_i) \to X_{i}^{ext}(t_i)), 1 \leq j \leq a_i \}$ are part of the cut. These contribute a total value of $a_i\alpha$.
		\item For any node $X_{i, j}^{out}(t_{i, j}), a_i + 1 \leq j \leq m$, if the corresponding in-node $X_{i, j}^{in}(t_{i, j})$ belongs to IFG$_S$, then the edge $(X_{i, j}^{in}(t_{i,j}) \to X_{i, j}^{out}(t_{i,j}))$ appears in the cut, and contributes a value of $\alpha$ to the cut. Now, consider the case when the in-node $X_{i, j}^{in}(t_{i, j})$ belongs to IFG$_T$. In this case, consider the following two sub cases:
		\begin{itemize}
			\item The node $X_{i, j}(t_{i, j})$ is not a replacement node: This means that, either the edge $(X_{i, j}^{out}(t_{i,j}-1) \to X_{i, j}^{in}(t_{i,j}))$ appears in the cut, if $t_{i,j} > 0$, or the edge $(S \to X_{i, j}^{in}(t_{i,j}))$ appears in the cut, if $t_{i,j} = 0$. In any case, the contribution to the overall value of the cut is at least $\alpha$.
			\item The node $X_{i, j}(t_{i, j})$ is a replacement node of $X_{i, j}(t_{i, j}-1)$: We know that $\ell$ local helper nodes and $d$ external nodes are involved in repair. It is straightforward to see that out of the $\ell$ local helper nodes, at most $(j-1)$ belong to IFG$_T$. To see this, note that the potential candidates for the local helper nodes that appear in IFG$_T$ correspond to the physical nodes\footnote{It may be noted that we count the physical nodes $X_{i, 1}, \ldots, X_{i, a_i}$ among the possible set of local helpers, although we assume that $X_{i, j}(t_i), 1 \leq j \leq a_i$ appears in IFG$_S$. This is because, we cannot discount the possibility that $X_{i, j}(t_{i,j'}-1)$ appears in IFG$_T$, for $j \leq a_i, j' > a_i$.} $X_{i, 1}, X_{i, 2}, \ldots, X_{i, j-1}$. The version of the physical node $X_{i, j'}, j' > j$ , if it aids in the repair process, appears in IFG$_S$ because of our definition of $X_{i,j'}(t_{i, j'})$. Next, note that out of the $d$ external nodes, at most $(i-1)$ belong to IFG$_T$. In this case, the contribution to the value of the cut, due to the edges that aid in repair, is lower bounded by $(\ell - j+1)^+\alpha + (d - (i-1))^+\beta$.
		\end{itemize}
	\end{itemize}
	Based on the observations above, the value of the cut is lower bounded by
	\begin{eqnarray}
		\text{mincut}(S-T) & \geq & \sum_{i = 1}^{k} \left( a_i\alpha  + \sum_{j = a_i + 1}^{m}\min(\alpha, (\ell - j+1)^+\alpha + (d - (i-1))^+\beta) \right) \label{eq:useinintra}\\
		& =  & a_ik\alpha + \sum_{i = 1}^{k} \sum_{j = a_i + 1}^{\ell} \alpha  +  \sum_{i = 1}^{k} \sum_{j = \max(\ell, a_i) + 1}^{m} \min(\alpha, (d - (i-1))^+\beta) \label{eq:useinintra1}\\
		& = & \max(a_i, \ell)k\alpha + (m - \max(a_i, \ell))\sum_{i = 1}^{k} \min(\alpha, (d - (i-1))^+\beta) \\
		& \geq & \ell k\alpha + (m - \ell)\sum_{i = 1}^{k} \min(\alpha, (d - (i-1))^+\beta),
	\end{eqnarray}
	for any $a_i, 0 \leq a_i \leq m$. This completes the proof of the converse.
\end{proof}

\vspace{0.1in}

\begin{defn}[Optimal Code]
	Code $\mathcal{C}_m$ is said to be optimal or capacity achieving, if its file-size $B = B^*$, where $B^*$ is as given in Theorem \ref{thm:file_size}.
\end{defn}

\subsection{Minimum Storage and Minimum Inter-Cluster Bandwidth Operating Points} \label{sec:tradeoff}
We now define the MSR and MBR operating points based on the bound in Theorem \ref{thm:file_size}, whenever $d > 0$. The operating point $d\beta = \alpha$ will be identified with the MBR operating point. An optimal code $\mathcal{C}_m$ at the MBR operating point will be referred to as an MBR code. To see the rationale behind the definition of the MBR operating point, we recall our assumption that  whenever $d > 0$, the encoding function does not introduce any dependency among the content of the nodes of a cluster.  In this case, assuming that $B$ symbols of the uncoded file are uniformly distributed over  $\mathbb{F}_q^{B}$, it is straightforward to see the necessity of the condition $d\beta \geq \alpha$ for any code $\mathcal{C}_m$.

Towards defining the MSR point, we note that $B \leq \ell k \alpha + (m -\ell) \min(d, k) \alpha$. The MSR code is an optimal code having $B = \ell k \alpha + (m -\ell) \min(d, k) \alpha$, and has the lowest possible inter-cluster repair bandwidth. It is clear that the parametrization of the MSR operating point depends on relation between $d$ and $k$.  For $d \geq k$, the optimal file-size of $B = mk\alpha$ is achievable only if $(d-k+1)\beta \geq \alpha$, and thus the operating point $(d-k+1)\beta = \alpha$ will be identified with the MSR operating point.  For $ 1 \leq d \leq k-1$, note that $B \leq \ell k\alpha + (m-\ell)d\alpha$, and equality is achieved only if $\beta \geq \alpha$. The operating point $\alpha = \beta$ will be identified with the MSR operating point for the case when $ 1 \leq d \leq k-1$.

Based on the definition of the MSR and MBR operating points, we note that for $d \geq k$, the range of $\alpha$ considered while plotting the trade-off is given by $(d-k+1)\beta \leq \alpha \leq d\beta$, and when $ 1 \leq d \leq k-1$, the range of $\alpha$ is given by $ \beta \leq \alpha \leq d\beta$.

\section{Code Constructions} \label{sec:code}

In this section, we describe our optimal code constructions. Two constructions are presented; the first one is an instance of an exact repair code, and results in optimal codes at the MSR and MBR points under the setting of generalized regenerating codes; the second construction is a functional-repair regenerating code. Both codes can withstand any number of repairs for the duration of operation of the system. The exact repair code withstands any number of repairs by definition, since after each repair the data on all nodes is the same as at the start of system operation. This logic does not hold for functional repair codes, because the repaired node content is generally different from the original one. Network-coding based achievability proofs for functional-repair work only if there is a known upper bound on the number of repairs that occur over the lifetime of the system. Our functional repair code relies on the construction in \cite{Wu_regen}, which allows our code to operate for arbitrarily many repairs. For both constructions, we rely on existing optimal classical regenerating codes that are linear. By a linear regenerating code, we mean that both encoding and repair are performed via linear combinations of either the input or the coded symbols, respectively. The first construction generates an optimal $(n, k, d) (\alpha, \beta)(m, \ell)$ code for any $m, \ell \leq m-1, 1 \leq d \leq n-1$, whenever an optimal $(n, k, \min(k, d))(\alpha, \beta)$ classical exact repair linear regenerating code exists. Our functional repair code construction is limited to the case $\ell = m-1, d \geq k$.

For a linear $(n, k, d) (\alpha, \beta)$ classical regenerating code that encodes a data file of size $B$ symbols, one can associate a generator matrix $G$ of size $B \times n\alpha$. Without loss of generality, the first $\alpha$ columns of $G$ generates the content of node $1$, and so on. We say that two $(n, k, d) (\alpha, \beta)$ classical linear regenerating codes $\mathcal{C}_1$ and $\mathcal{C}_2$, having generator matrices $G_1$ and $G_2$ are identical, if $G_1 = G_2$. Since we assume that repair is also a linear operation, it follows that the set of linear functions that define repair operations can be taken the same for  $\mathcal{C}_1$ and $\mathcal{C}_2$ that are identical.  The following lemma will be used while we prove  optimality of the exact-repair construction. The proof of the lemma is straightforward, and is omitted.

\vspace{0.1in}

\begin{lem} \label{lem:regen_combo}
	Let $\{\mathcal{C}_i, 1 \leq i \leq s, s \geq 1\}$ denote identical $(n, k, d) (\alpha, \beta)$ classical  exact repair linear regenerating codes.  Let ${\bf c}_i \in \mathbb{F}_q^{n\alpha}$ denote a generic codeword of $\mathcal{C}_i, 1 \leq i \leq s$, where the first $\alpha$ symbols of ${\bf c}_i$  are the content of node $1$, and so on. Suppose that we are given $a_i \in \mathbb{F}_q, 1 \leq i \leq s$ such that not every $a_i$ is $0$, define a new $n$-length array code $\mathcal{C}$ over $\mathbb{F}_q^{\alpha}$ as $\mathcal{C} = \{ \sum_{i = 1}^{s}a_i {\bf c}_i,  {\bf c}_i \in \mathcal{C}_i\}$. Then, the code $\mathcal{C}$ is an $(n, k, d) (\alpha, \beta)$ classical  exact repair linear regenerating code over $\mathbb{F}_q$, and is identical to $\mathcal{C}_i, 1 \leq i \leq s$.
\end{lem}

\vspace{0.1in}

\subsection{Exact Repair Code Construction}

We begin with a description of the code, and then show its data collection and repair properties. The construction itself is a generalization of the example presented in  Section \ref{sec:ex}.

\vspace{0.1in}

\begin{constr} \label{constr:exact}
	Let $\mathcal{C}_j, 1 \leq j \leq \ell$ denote $[n, k]$ MDS array codes over $\mathbb{F}_q^{\alpha}$. The amount of data that can be encoded with these $\ell$ codes is $\ell k \alpha$. Next, let $\mathcal{C}_{j}, \ell + 1 \leq j \leq m$ denote $(n, k, d' = \min(d, k))(\alpha, \beta)$ classical exact repair linear regenerating codes (also over $\mathbb{F}_q$), each having a file size $B' = \sum_{i = 0}^{k - 1}\min(\alpha, (d' - i)\beta)$. We require that all the codes $\mathcal{C}_{j}, \ell + 1 \leq j \leq m$ are identical. For encoding, we first divide the data file of size $B^* = \ell k \alpha + (m - \ell)B'$ into $m$ stripes, such that first $\ell$ have size $k \alpha$, and the last $m-\ell$ have size $B'$. Stripe $j, 1  \leq j \leq m$ is encoded by $\mathcal{C}_j$ to generate the coded symbols  ${\bf c}_j = [c_{1, j}, c_{2, j}, \ldots, c_{n\alpha, j}]^T$.  Next, consider an $m \times m$ invertible matrix $A$ over $\mathbb{F}_q$ such that the first the $\ell$ rows of $A$ generate an $[m, \ell]$ MDS code $\mathbb{F}_q$. Let matrix $A$ be decomposed as
	\begin{eqnarray} \label{eq:matrixA}
		A_{m \times m } & = & \left[ \begin{array}{c} E_{\ell \times m} \\ \hline \\ F_{m - \ell \times m}\end{array} \right].
	\end{eqnarray}
	Thus, the matrix $E_{\ell \times m}$ generates an $[m, \ell]$ MDS code. The coded data stored in the various clusters is generated as follows:
	\begin{eqnarray} \label{eq:encode}
		[ {\bf c}'_{1} \ {\bf c}'_{2} \ \cdots \ {\bf c}'_{m} ] & = & [ {\bf c}_{1} \ {\bf c}_{2} \ \cdots \ {\bf c}_{m} ]A_{m \times m}.
	\end{eqnarray}
	The content of node $j$ in cluster $i$ is given by $[c'_{(i-1)\alpha + 1, j}, c'_{(i-1)\alpha + 2, j},\ldots, c'_{i\alpha, j}]^T, 1 \leq i \leq n, 1 \leq j \leq m$. This completes the description of the construction. A pictorial overview of the description appears in Fig. \ref{fig:code_encode}.
\end{constr}
\begin{figure*}
	\centering
	\includegraphics[height=2.5in]{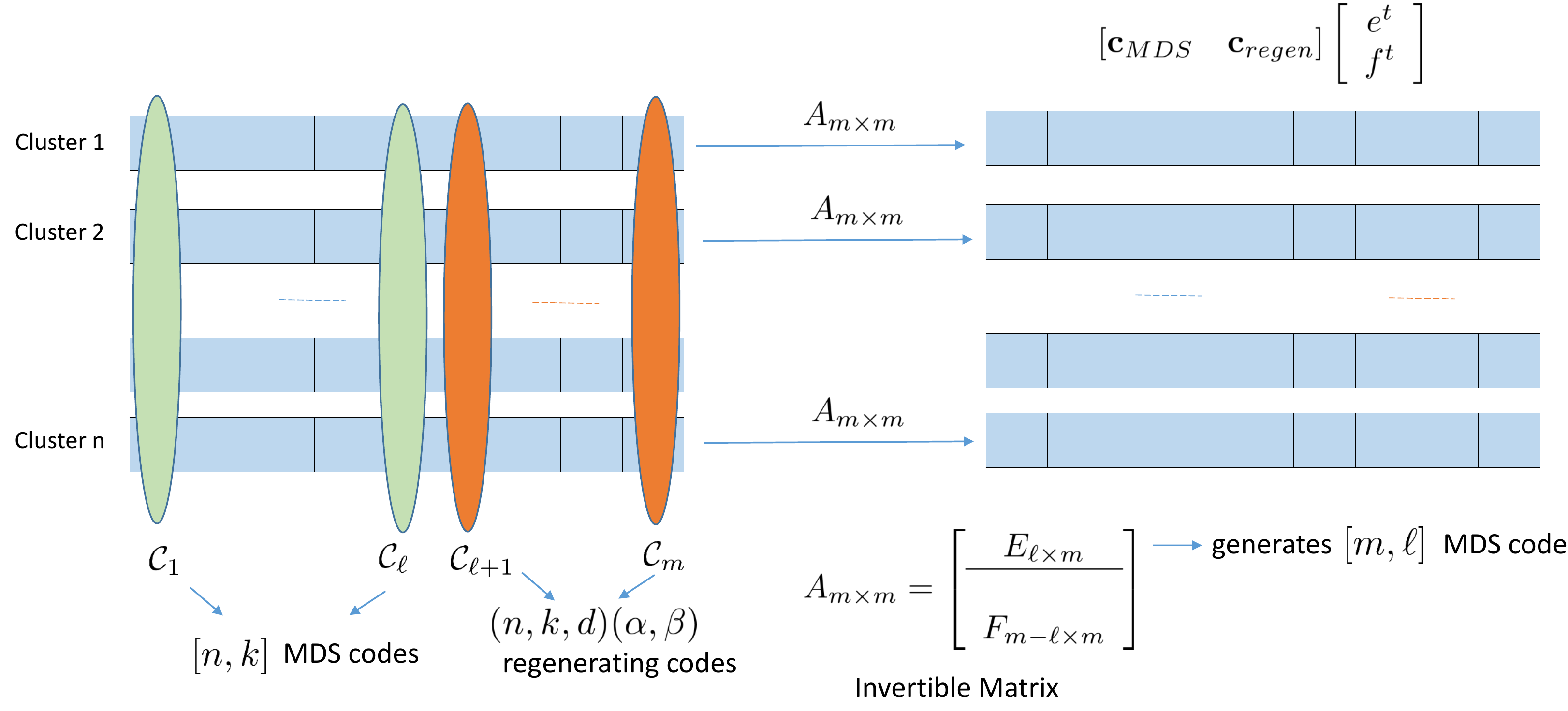}
	\caption{Illustration of the exact repair code construction. We first stack $\ell$ MDS codes and $(m - \ell)$ classical regenerating codes, and then transform each row via the invertible matrix $A$. The first $\ell$ rows of the matrix $A$ generates an $[m, \ell]$ MDS code.}
	\label{fig:code_encode}
\end{figure*}

We next prove optimality property of Construction \ref{constr:exact}.

\vspace{0.1in}

\begin{thm} \label{thm:exact_code}
	The code described in Construction \ref{constr:exact} is an optimal exact repair generalized regenerating code, for any $m, \ell \leq m$. The optimal code can be constructed whenever an optimal $(n, k, d'=\min(d, k)) (\alpha, \beta)$ exact repair linear regenerating code exists, having a file size $B' = \sum_{i = 0}^{k - 1}\min(\alpha, (d' - i)\beta)$.
\end{thm}
\begin{proof}
	It is clear that the code in Construction \ref{constr:exact} has a file size $B^*$, where $B^*$ is as given in Theorem \ref{thm:file_size}. Further, the data collection property of the code is also straightforward to check, and this essentially follows from the facts that $1)$ the matrix $A$ is invertible, and $2)$ each of the codes $\mathcal{C}_i, 1 \leq i \leq m$ is uniquely decodable given its coded data belonging to any $k$ clusters. Towards examining the repair property of the code let us rewrite \eqref{eq:encode} as follows:
	\begin{eqnarray}
		[ {\bf c}'_{1} \ {\bf c}'_{2} \ \cdots \ {\bf c}'_{m} ] & = & [ {\bf c}_{1} \ {\bf c}_{2} \ \cdots \ {\bf c}_{m} ]A_{m \times m} \\
		& = & \left[ {\bf C}_{MDS} \ {\bf C}_{regen}  \right ]A_{m \times m} \\
		& = & \left [ \begin{array}{cc}  {\bf c}^{(1)}_{MDS} & {\bf c}^{(1)}_{regen} \\
			{\bf c}^{(2)}_{MDS} & {\bf c}^{(2)}_{regen} \\
			\vdots & \vdots \\
			{\bf c}^{(n)}_{MDS} & {\bf c}^{(n)}_{regen}
		\end{array}  \right ]A_{m \times m}, \label{eq:encode_2}
	\end{eqnarray}
	where ${\bf C}_{MDS} = [{\bf c}_{1} \ \cdots \ {\bf c}_{\ell} ]$ and ${\bf C}_{regen} = [{\bf c}_{\ell+1} \ \cdots \ {\bf c}_{m} ]$. The matrices ${\bf c}^{(i)}_{MDS}$ and ${\bf c}^{(i)}_{regen}, 1 \leq i \leq n$ denote rows $(i-1)\alpha + 1, \ldots, i\alpha$ of ${\bf C}_{MDS}$ and ${\bf C}_{regen}$, respectively. Let us also expand the decomposition of matrix $A$ in \eqref{eq:matrixA} further as follows:
	\begin{eqnarray}
		A_{m \times m } & = & \left[ \begin{array}{c} E_{\ell \times m} \\ \hline \\ F_{m - \ell \times m}\end{array} \right] \\
		& = &   \left[  \begin{array}{cccc} {\bf e}_1^T  & {\bf e}_2^T & \cdots & {\bf e}_m^T \\ \hline \\ {\bf f}_1^T  & {\bf f}_2^T & \cdots & {\bf f}_m^T \end{array} \right], \label{eq:matrixA_2}
	\end{eqnarray}
	where ${\bf e}_j^T$ and  ${\bf f}_j^T, 1 \leq j \leq m$ denote the $j^{\text{th}}$ column of the matrices $E$ and $F$, respectively. Based on \eqref{eq:encode_2} and \eqref{eq:matrixA_2}, it can be seen that the content of node $j$ in cluster $i$ is given by
	\begin{eqnarray} \label{eq:data_clusteri}
		\left[{\bf c}^{(i)}_{MDS} \  \  {\bf c}^{(i)}_{regen}\right]\left[  \begin{array}{c} {\bf e}_j^T  \\ \hline \\ {\bf f}_j^T   \end{array} \right]
	\end{eqnarray}
	Given the notation above, without loss of generality, consider repairing node $\ell + 1$ in cluster $1$ with the help of $1)$ the first $\ell$ local nodes in cluster $1$ and $2)$ clusters $2, \ldots, d'+1$. Let us first examine the role of the $\ell$ local nodes in the repair process. Let $E'$ and $F'$ denote the first $\ell$ columns of $E$ and $F$, respectively. By assumption, $E$ generates an $[m, \ell]$ MDS code, and hence the submatrix $E'$ is invertible. In this case, the content from the $\ell$ local nodes can be put together to generate
	\begin{eqnarray} \label{eq:local_data}
		\left( \left[{\bf c}^{(1)}_{MDS} \  \  {\bf c}^{(1)}_{regen}\right]\left[  \begin{array}{c} E' \\ \hline \\ F' \end{array} \right] \right) E'^{-1} {\bf e}_{\ell + 1}^T & = & \left[{\bf c}^{(1)}_{MDS} \  \  {\bf c}^{(1)}_{regen}\right]\left[  \begin{array}{c} {\bf e}_{\ell + 1}^T  \\ \hline \\ {\bf \widehat{f}}_{\ell + 1}^T   \end{array} \right],
	\end{eqnarray}
	where ${\bf \widehat{f}}_j^T = F'E'^{-1}{\bf e}_{\ell + 1}^T$. Thus, the local helper nodes serve to recover the part corresponding to the MDS-codes' components given by ${\bf c}^{(1)}_{MDS}{\bf e}_{\ell + 1}^T$. However, the regenerating-codes' components ${\bf c}^{(1)}_{regen}{\bf \widehat{f}}_{\ell + 1}^T$ differs from the original ${\bf c}^{(1)}_{regen}{\bf {f}}_{\ell + 1}^T$.
	
	Let us next examine the role of the $d'$ remote helper clusters. We know that the data stored in cluster $i$ is given by $\left[{\bf c}^{(i)}_{MDS} \  \  {\bf c}^{(i)}_{regen}\right]A$. Since matrix $A$ is invertible, the vector ${\bf c}^{(i)}_{regen}$ can be recovered from this.
	Now, if we define the code $\widehat{\mathcal{C}}_i = \left\{ \sum_{j=1}^{m-\ell}(f_j - \widehat{f}_j) {\bf c}_{\ell + j}, {\bf c}_{\ell + j} \in  \mathcal{C}_{\ell + j} \right\}$, we know from Lemma \ref{lem:regen_combo} that $\widehat{\mathcal{C}}_i$ is an $(n, k, d')(\alpha, \beta)$ classical exact repair linear regenerating code, which is identical to $\mathcal{C}_j, \ell + 1 \leq j \leq m $. Thus, cluster $i, 2 \leq i \leq d' + 1$ generates and passes the helper data ($\beta$ symbols) toward the repair of the first vector symbol for the code $\widehat{\mathcal{C}}_i$. The replacement node regenerates ${\bf c}^{(1)}_{regen}({\bf f}_j^T - \widehat{{\bf f}}_j^T)$ using the helper data from the $d'$ remote clusters, and combines it with the local helper data (see \eqref{eq:local_data}) to correct the regenerating-codes' components, and restore the content of the lost node. A pictorial illustration of the repair process is shown in Fig. \ref{fig:code_repair}.
\end{proof}

\begin{figure*}
	\centering
	\includegraphics[height=2.5in]{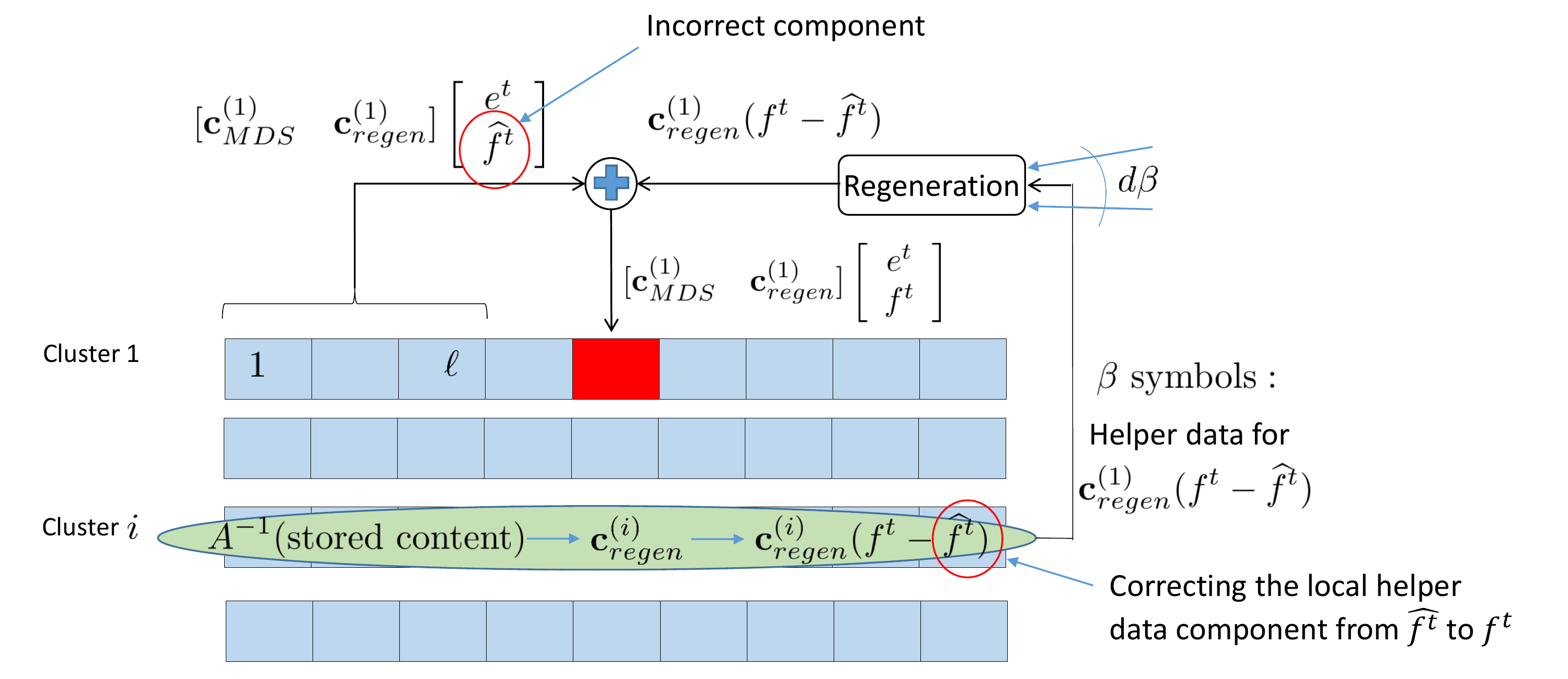}
	\caption{An illustration of the node repair process for exact repair generalized regenerating code obtained in Construction \ref{constr:exact}.}
	\label{fig:code_repair}
\end{figure*}

\subsection{A Functional Repair Code for Arbitrary Number of Failures}

In this section, we show the existence of optimal functional repair codes over a finite field that can tolerate an arbitrary number of repairs for the duration of operation of the system. We show the existence for any $(n,k, d)(\alpha, \beta)(m, \ell = m-1)$. The code construction is similar to the one used in the example in Section \ref{sec:ex},  and combines $m-1$ MDS array codes $\mathcal{C}_1, \ldots, \mathcal{C}_{m-1}$ with an $(n,k, d)(\alpha, \beta)$ functional repair code $\mathcal{C}_m$ for the classical setting. The code $\mathcal{C}_m$ is one that can tolerate an arbitrary number of repairs for the duration of operation of the system. The following lemma is a direct consequence  of Theorem $3$ and the description in Section $V$ of \cite{Wu_regen}, and guarantees the existence of the code $\mathcal{C}_m$ that we use here.

\vspace{0.1in}

\begin{lem} \label{lem:wu_code}
	For any $(n, k, d)(\alpha, \beta)$, there exists an optimal deterministic classical functional repair linear  regenerating code over $\mathbb{F}_q$ that can tolerate an arbitrary number of repairs for the duration of operation of the system, whenever $q > q_0$, where $q_0$ is entirely determined by the parameters $(n, k, d)(\alpha, \beta)$, and is independent of the number of repairs performed over the lifetime of the code.
\end{lem}

\vspace{0.1in}

In the above lemma, by a deterministic regenerating code, we mean that the regenerated data corresponding to a repair operation of a given physical node is uniquely determined given the content of the helper nodes.  As we shall see, the fact the code is deterministic is important to ensure the data collection property of our functional repair construction.

Below, we first describe the code construction, along with the repair procedure, and then show the  optimality property of the code.  In the following construction, we assume that the characteristic of the finite field $\mathbb{F}_q$ is $2$; i.e., $q = 2^w$ for some $w > 0$. The assumption is made only for the ease of description, the construction can be modified to accommodate finite fields of any characteristic.

\vspace{0.1in}

\begin{constr} \label{constr:functional}
	Let $\mathcal{C}_j, 1 \leq j \leq m-1$ denote $[n, k]$ MDS array codes over $\mathbb{F}_q^{\alpha}$. The amount of data that can be encoded with these $m$ codes is $(m-1) k \alpha$. Next, let $\mathcal{C}_{m}$ denote an $(n, k, d' = \min(d, k))(\alpha, \beta)$ classical functional repair linear regenerating code whose existence is guaranteed by Lemma \ref{lem:wu_code}. The code $\mathcal{C}_{m}$ has a file size $B' = \sum_{i = 0}^{k - 1}\min(\alpha, (d' - i)\beta)$. For encoding, we first divide the data file of size $B^* = \ell k \alpha + (m - \ell)B'$ into $m$ stripes, such that first $m-1$ have size $k \alpha$, and the last one has size $B'$. Stripe $j, 1  \leq j \leq m$ is encoded by $\mathcal{C}_j$ to generate the coded symbols  ${\bf c}_j = [c_{1, j}, c_{2, j}, \ldots, c_{n\alpha, j}]^T$. The arrangement of coded data in the various nodes is identical to that in the example in Section \ref{sec:ex}. Thus, node $j, 1 \leq j \leq m-1$ in cluster $i$ stores the vector $[c_{(i-1)\alpha + 1, j}, c_{(i-1)\alpha + 2, j},\ldots, c_{i\alpha, j}]^T$. Node $m$  in cluster $i$ stores the sum of the content of  $m-1$  nodes with symbols  $[c_{(i-1)\alpha + 1, m}, c_{(i-1)\alpha + 2, m},\ldots, c_{i\alpha, m}]^T$ of the regenerating code $\mathcal{C}_{m}$, i.e., content of node $m$ is given by $[\sum_{j = 1}^{m}c_{(i-1)\alpha + 1, j}, \sum_{j = 1}^{m}c_{(i-1)\alpha + 2, j},\ldots, \sum_{j = 1}^{m}c_{i\alpha, j}]^T$.
	This completes the description of the initial layout of the coded data. Since the code is a functional repair code, the code description is not complete unless we specify the procedure for node repair, as well. We do this next.
	
	\emph{\underline{Node Repair}:} Let $Y_{i,j}(t) \in \mathbb{F}_q^{\alpha}$ denote the content of node $j$ in cluster $i$, after the $t^{\text{th}}, t \geq 0$ repair, anywhere in the system. The quantities $\{Y_{i,j}(0), 1 \leq i \leq n, 1 \leq j \leq m\}$ denote the initial content present in the system, and are as described above. The repair procedure is such that the vector $[\sum_{j = 1}^{m}Y_{1,j}(t), \sum_{j = 1}^{m}Y_{2,j}(t),\ldots, \sum_{j = 1}^{m}Y_{n,j}(t)]^T$ remains as a valid codeword of the functional repair regenerating code $\mathcal{C}_{m}$, for every $t \geq 0$ (to be proved in Theorem \ref{thm:func_code}). Clearly, the above statement is true for $t = 0$. The repair procedure can be described recursively as follows: Let the $t^{\text{th}}$ repair be associated with node $i'$ in cluster $j'$. Each of the $d'$ remote helper clusters, say $i$, internally computes  $\sum_{j = 1}^{m}Y_{i,j}(t-1)$, and passes the $\beta$ symbols toward the repair of  $\sum_{j = 1}^{m}Y_{i',j}(t-1)$. The replacement node first of all regenerates $\widehat{Y}_{i'}(t-1)$, as the replacement of $\sum_{j=1}^{m}Y_{i',j}(t-1)$, given the  helper data from the $d$ remote clusters. Next, since $\ell = m-1$, the replacement node gets access to local helper data $\{Y_{i',j}(t-1), 1 \leq j \leq m, j \neq j'\}$. The content that is eventually stored in the replacement node is computed as follows:
	\begin{eqnarray} \label{eq:replacednode}
		Y_{i', j'}(t)  & = & \sum_{\substack{j = 1 \\ j \neq j'}}^{m}Y_{i', j}(t-1)  \ + \ \widehat{Y}_{i'}(t-1).
	\end{eqnarray}
	Also, for any $(i, j) \neq (i', j')$, we assume that
	\begin{eqnarray}  \label{eq:othernodes}
		Y_{i, j}(t) = Y_{i, j}(t-1).
	\end{eqnarray}
	This completes the description of the repair process and the code construction.
\end{constr}

\vspace{0.1in}

In the following theorem, we argue the optimality property of the above construction. Specifically, we show that the code retains the functional repair and data collection properties, after every repair. We assume that the data collector is aware of the entire repair-history of the system. By this we mean that the data collector is aware of $1)$ the exact sequence of $t$ failures and repairs that has happened in the system, and $2)$ the indices of the remote helper clusters that aided in each of the $t$ repairs.

\vspace{0.1in}

\begin{thm}\label{thm:func_code}
	The code described in Construction \ref{constr:functional} is an optimal $(n, k, d)(\alpha, \beta)(m, \ell = m-1)$ functional repair generalized regenerating code.
\end{thm}
\begin{proof}
	It is clear that the code in Construction \ref{constr:functional} has a file-size $B^*$, as given by Theorem \ref{thm:file_size}. Toward showing that the code retains functional repair property, it is sufficient if we show  that the vector $[\sum_{j = 1}^{m}Y_{1,j}(t),$ $ \sum_{j = 1}^{m}Y_{2,j}(t),\ldots, \sum_{j = 1}^{m}Y_{n,j}(t)]^T$ remains as a valid codeword of the functional repair regenerating code $\mathcal{C}_{m}$, for every $t \geq 0$. We do this inductively. Clearly, the statement is true for $t=0$. Let us next assume that the statement is true for $t  = t' \geq 0$, and show its validity for $t = t' + 1$. Assume that the $(t' + 1)^{\text{th}}$ repair is associated with node $j'$ in cluster $i'$. The relation between the content of the various nodes before and after the $(t' + 1)^{\text{th}}$ repair are obtained via \eqref{eq:replacednode} and \eqref{eq:othernodes}. In this case, the quantities $\{\sum_{j = 1}^{m}Y_{i,j}(t'+1), 1 \leq i \leq n\}$ are given by
	\begin{eqnarray}
		\sum_{j = 1}^{m}Y_{i',j}(t'+1) & \stackrel{(a)}{=} &  \widehat{Y}_{i'}(t'), \\
		\sum_{j = 1}^{m}Y_{i,j}(t'+1) & = & \sum_{j = 1}^{m}Y_{i,j}(t'), \  1 \leq i \leq n, \ i \neq i',
	\end{eqnarray}
	where $(a)$ follows from our assumption that the finite field $\mathbb{F}_q$ has characteristic $2$. Now, recall that $\widehat{Y}_{i'}(t')$ is the replacement of $\sum_{j = 1}^{m}Y_{i',j}(t')$,  which is regenerated using the helper data generated using $d$ elements of the set $\{\sum_{j = 1}^{m}Y_{i,j}(t'), 1 \leq i \leq n, i \neq i'\}$. Combining with the induction hypothesis for $t = t'$, it follows that the induction statement holds good for $t = t'+1$ as well. This completes the proof of functional repair property of the code.
	
	Let us next see how data collection is accomplished after $t, t \geq 0$ repairs in the system. Without loss of generality assume that a data collector connects to  clusters $1, 2, \ldots, k$, and accesses $\{Y_{i, j}(t), 1 \leq i \leq k, 1 \leq j \leq m\}$. The data collector as a first step computes the vector $[\sum_{j = 1}^{m}Y_{1,j}(t), \sum_{j = 1}^{m}Y_{2,j}(t),\ldots, \sum_{j = 1}^{m}Y_{k,j}(t)]^T$, and uses this to decode the data corresponding to the code $\mathcal{C}_m$. Now, recall the fact that the code $\mathcal{C}_m$ is deterministic, and also our assumption that the data collector is aware of the entire repair-history of the system. In this case, having decoded $\mathcal{C}_m$, using \eqref{eq:replacednode} and \eqref{eq:othernodes}, the data collector can iteratively recover $\{Y_{i, j}(t'), 1 \leq k \leq, 1 \leq j \leq m\}$, for $ t \geq t' \geq 0$ by starting at $t' = t$ and proceeding backwards until the content at $t' = 0$ is recovered (essentially, we are rewinding the system by eliminating the effects of all the repairs, starting from the last one and proceeding backwards in time). Finally, from Construction \ref{constr:functional}, we know that the content $\{Y_{i, j}(0), 1 \leq k \leq, 1 \leq j \leq m-1\}$ is the stacked coded data corresponding to the $m-1$ $[n, k]$ MDS coded $\mathcal{C}_1, \ldots, \mathcal{C}_{m-1}$, and thus these codes can also be decoded. The completes the proof of data collection, and also the theorem.
	
\end{proof}

\section{Intra-cluster bandwidth for Optimal Codes} 
\label{sec:intra_cluster_bandwidth}

We now turn our attention to calculate the amount of intra-cluster repair bandwidth that is needed for a ${(n, k, d)(\alpha, \beta)(m, \ell)}$ code $\mathcal{C}_m$ to have optimal file-size $B^*$, given by Theorem \ref{thm:file_size}. As discussed in Section \ref{sec:sys_model}, there are two contributors to intra-cluster repair bandwidth: $1)$ the local helper bandwidth $\gamma$, which is the amount of data that each of the local helper nodes contributes to repair, and $2)$ the remote helper bandwidth $\gamma'$, which is amount of the data  that each of the $\ell'$ nodes of a remote helper cluster contribute toward computing the $\beta$ symbols of the cluster. In this section, we   study individually the minimum requirements on the parameters $\gamma$, $\gamma'$ and $\ell'$. For obtaining lower bound on $\gamma$, we continue to work with the IFG model in Section \ref{sec:IFG_inter}, except for the fact that links that connect the local helper out-nodes to the in-node of the replacement node, will have a capacity $\gamma$, instead of $\alpha$.  The IFG model in Section \ref{sec:IFG_intra} will be used when we compute lower bound on $\gamma'$ and $\ell'$. We also prove the tightness of the bound on $\gamma$ via a converse; however, no such converse is known to us regarding the bound on $\gamma'$. We note that, while computing the bound on $\gamma$, we ignore the effects of limited $\gamma'$ and $\ell'$ (and vice versa), i.e., we assume that $\ell' = m$ and $\gamma' = \alpha$. Further, the bounds on $\gamma$ and $\gamma'$ (or $\ell'$) are obtained under the assumptions that $d > 0$ and $d \geq k$, respectively.

\subsection{Bound on Local Helper Node Intra-cluster Repair Bandwidth, $\gamma$}

\begin{figure}
	\centering
	\includegraphics[height=2.3in]{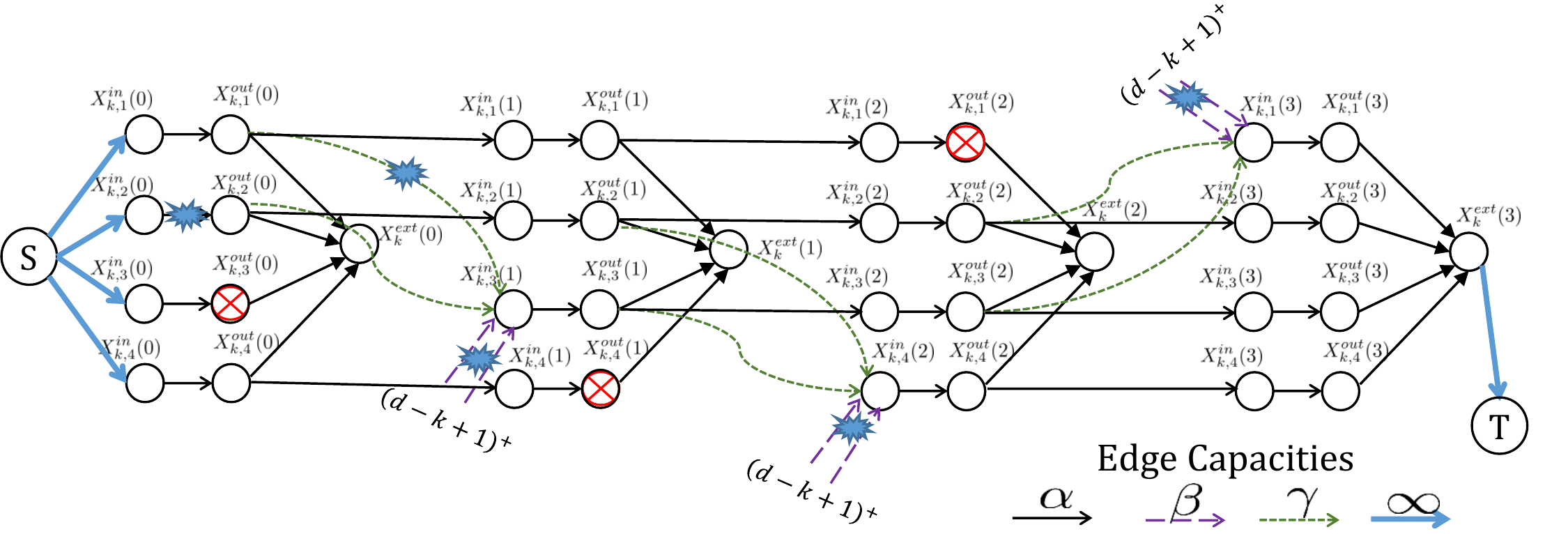}
	\caption{An illustration of the evolution of the $k$-th cluster of the information flow graph used in cut-set based lower bound for $\gamma$ in Theorem \ref{thm:local_BW}. In this figure, we assume that $m = 4, \ell = 2$. Nodes $3, 4, 1$ fail in this respective order. For the repair of node $3$, nodes $1$ and $2$ act as the local helper nodes. For the repair of the remaining two nodes, nodes $2$ and $3$ act as the local helper nodes. Also indicated is our choice of the $S$-$T$ cut used in the bound derivation.}
	\label{fig:cut_gamma}
	\vspace{-10pt}
\end{figure}

\begin{thm} \label{thm:local_BW}
	For an optimal functional repair generalized regenerating code with parameters  $\{(n, k, d > 0),$ $ (\alpha, \beta), (m, \ell)\}, \gamma' = \alpha, \ell' = m$, local helper node bandwidth $\gamma$ is lower-bounded by
	\begin{eqnarray}
		\gamma & \geq \gamma^* \triangleq & \alpha - (d-k+1)^+\beta.
	\end{eqnarray}
	Further, if there is a known upper bound on the number of repairs that occur over the life-time of the system, the above bound is sharp; i.e., the functional repair capacity of the system remains as $B^*$ as long as $\gamma \geq \gamma^*$.
\end{thm}

\begin{proof}
	For the bound, we consider a system evolution similar to that used in proof of Theorem \ref{thm:file_size}, and demonstrate a cut-set whose value depends on $\gamma$. The lower bound on $\gamma$ follows from  the observation that the value of this cut is necessarily lower bounded by $B^*$ for a capacity-achieving code. We shall then prove that, as long as $\gamma \geq \gamma^*$, the min-cut of any valid IFG is necessarily lower bounded by $B^*$; in this case, like in the proof of Theorem \ref{thm:file_size}, we know that the functional repair capacity remains as $B^*$, as long as there is a known upper bound on the number of repairs in the system. We start with the proof of the lower bound. Consider the same system evolution as in the proof of the bound in Theorem \ref{thm:file_size}, except for the $k$-th cluster accessed by the data collector. Thus, physical nodes $X_{i, \ell + 1}, X_{i,\ell+2}, \ldots X_{i,m}$ fail in this order in cluster $i = 1$, then in cluster $i=2$, and so on, until cluster $i=k-1$. Note that each of the first $k-1$ clusters experiences a total of $m - \ell$ node failures. For cluster $k$, we consider failure of $m - \ell + 1$ nodes, corresponding to physical nodes $X_{k, \ell + 1}, X_{k,\ell+2}, \ldots X_{k,m}, X_{k,1}$ in this respective order. In terms of the notation  introduced in \ref{sec:IFG_inter}, the sequence of failures in the $k^{\text{th}}$ cluster correspond to IFG nodes $X_{k, \ell + 1}(0), X_{k,\ell+2}(1), \ldots, X_{k,m}(m - \ell-1), X_{k,1}(m - \ell)$. For the repair of $X_{k, \ell + 1}(0)$, the local helper nodes used are $X_{k, 1}(0), X_{k, 2}(0), \ldots, X_{k, \ell}(0)$. For the repair of any of the remaining nodes $X_{k, (\ell + t) \mod (m)  + 1}(t), 1 \leq t \leq m-\ell$, the local helper nodes used are $X_{k, 2}(t), X_{k, 3}(t), \ldots, X_{k, \ell + 1}(t)$. Also, clusters $\mathcal{X}_1(m - \ell), \mathcal{X}_2(m - \ell), \ldots, \mathcal{X}_{\min(d, k - 1)}(m - \ell)$ are included in the set of remote clusters that aid in the repair of the $m - \ell + 1$ nodes in the $k^{\text{th}}$ cluster. An illustration of the IFG, for the $k^{\text{th}}$ cluster is shown in Fig. \ref{fig:cut_gamma}. Note in this figure that the edges corresponding to local help have capacity $\gamma$.
	
	Let data collector $T$ connect to clusters $\mathcal{X}_1(m-\ell), \ldots, \mathcal{X}_{k-1}(m-\ell), \mathcal{X}_{k}(m-\ell+1)$. Consider an $S$-$T$ cut in the IFG that partitions the graph nodes in clusters $1,\cdots,k-1$ in the same way as in \Cref{thm:file_size}; however it differs in the way the nodes of cluster $k$ are partitioned. The overall set of edges in the cut-set is given as below: \newline
	\emph{Clusters $1, \ldots, k-1$}:
	\begin{itemize}
		\item $\{(X_{i,j}^{in}(0) \to X_{i,j}^{out}(0), i\in [k-1], j\in [\ell]\}$.  Total capacity of these edges is $(k-1)l\alpha$.
		\item For each $ i\in [k-1], t \in [m-\ell]$, either the set of edges $\{(X_{i'}^{ext}(0) \to X_{i,\ell+ t}^{in}(t)), i'\in \{$remote helper cluster indices for the replacement node $X_{i,\ell+ t}^{in}(t) \} \backslash [\min\{i-1, d\}]$ or the  edge $ (X_{i,\ell + t}^{in}(t) \to  X_{i,\ell+t}^{out}(t))$. Between the two possibilities, we pick the one which has smaller sum-capacity. In this case, the total capacity of this part of the cut is given by $\sum_{i=1}^{k-1} \sum_{j=\ell+1}^{m} \min \{\alpha, (d-\min\{i-1, d\})\beta \} = (m - \ell)\sum_{i=1}^{k-1} \min \{\alpha, (d-i+1)^+\beta \}$.
	\end{itemize}
	\emph{Cluster $k$}:
	\begin{itemize}
		\item $(X_{k,1}^{out}(0)\to {X}^{in}_{k,\ell+1}(1))$ of capacity $\gamma$.
		\item $(X_{k,j}^{in}(0) \to X_{k,j}^{out}(0)), \forall j\in [2,\ell]$. 		Total capacity of these edges is $(\ell-1)\alpha$.
		\item Either the set of edges $\{(X_{i'}^{ext}(0) \to X_{k,(\ell + t) \mod (m)  + 1)}^{in}(t + 1)), i'\in \{$remote helper cluster indices for the replacement node $X_{k,(\ell + t) \mod (m)  + 1)}^{in}(t + 1)\} \backslash [\min\{i-1, d\}], 0 \leq t \leq (m-\ell)\}$ or the set of edges
		$ \{ (X_{k,(\ell + t) \mod (m)   + 1)}^{in}(t + 1) \to  X_{k,(\ell + t) \mod (m)   + 1)}^{out}(t + 1)), 0 \leq t \leq (m -\ell) \}$. Among the two sets, we pick the one which has smaller sum-capacity. In this case, the total capacity of these edges is $(m-\ell+1)\min\{\alpha, (d-k+1)^+ \beta\}$.
	\end{itemize}
	The value (say, $C_{cut}$) of the cut is given by
	\begin{align*}
		C_{cut}&=(k-1)\ell \alpha + (m - \ell)\sum_{i=0}^{k-2} \min \{\alpha, (d-i)^+\beta \}  + \gamma
		+ (\ell-1)\alpha + (m-l+1) \min\{\alpha, (d-k+1)^+ \beta\} \\
		&= k\ell \alpha + (m - \ell)\sum_{i=0}^{k-1} \min \{\alpha, (d-i)^+\beta \} -\alpha + \min\{\alpha, (d-k+1)^+ \beta\} + \gamma \\
		&= B^* -\alpha + \min\{\alpha, (d-k+1)^+ \beta\} + \gamma.
	\end{align*}
	Since we assume an optimal code, it must be true that $C_{cut} \geq B^*$, which results in $\gamma \geq \alpha - \min\{ \alpha, (d-k+1)^+\beta\}$. Finally, note that
	$\alpha - \min\{ \alpha, (d-k+1)^+\beta\} = \gamma^*$ in the two cases $d < k$, and $d \geq k$. The equivalence for the case $d \geq k$ follows since we only consider $\alpha \geq (d-k+1)\beta$, when $d \geq k$ (see Section \ref{sec:tradeoff}).
	
	We next prove the converse, where we show that,  as long as $\gamma \geq \gamma^*$, the min-cut of any valid IFG is necessarily lower bounded by $B^*$. Toward this, consider the proof of converse part of Theorem \ref{thm:file_size}, where we obtained a lower bound on the min-cut of any valid IFG. One can repeat the sequence of arguments exactly as in the proof of converse part of Theorem \ref{thm:file_size}, except with the change that the edges corresponding to local help have capacity $\gamma$ (instead of $\alpha$).  In this case, it can be seen that instead of \eqref{eq:useinintra}, we get the following lower bound on min-cut:
	\begin{eqnarray} \label{eq:gamma_1}
		\text{mincut}(S-T) & \geq & \sum_{i = 1}^{k} \left( a_i\alpha  + \sum_{j = a_i + 1}^{m}\min(\alpha, (\ell - j+1)^+\gamma + (d - (i-1))^+\beta) \right)
	\end{eqnarray}
	In the above expression, observe that if $\gamma \geq \gamma^*$, we have
	\begin{eqnarray}
		(\ell - j+1)^+\gamma + (d - (i-1))^+\beta & \geq & \alpha,
	\end{eqnarray}
	whenever $j \leq \ell, i \leq k$. In this case, it follows  that \eqref{eq:gamma_1} can be written as \eqref{eq:useinintra1}. It is then clear that  $\text{mincut}(S-T)$  is indeed lower bounded by $B^*$ as long as $\gamma \geq \gamma^*$. This completes the proof of the converse, and also the theorem.
\end{proof}

\subsection{Bounds on $\ell', \gamma'$}

In this section, we provide bounds on the parameters $\gamma'$ and $\ell'$. We use the second IFG model in Section \ref{sec:IFG_intra} here. Recall that in our setting, for any of the remote helper clusters, we allow any subset of $\ell'$ nodes in the cluster to be used to generate the $\beta$ symbols contributed by the cluster. Also, in this section, we make the assumption that the number of remote helper clusters $d \geq k$.

\begin{figure*}
	\centering
	\includegraphics[height=3in]{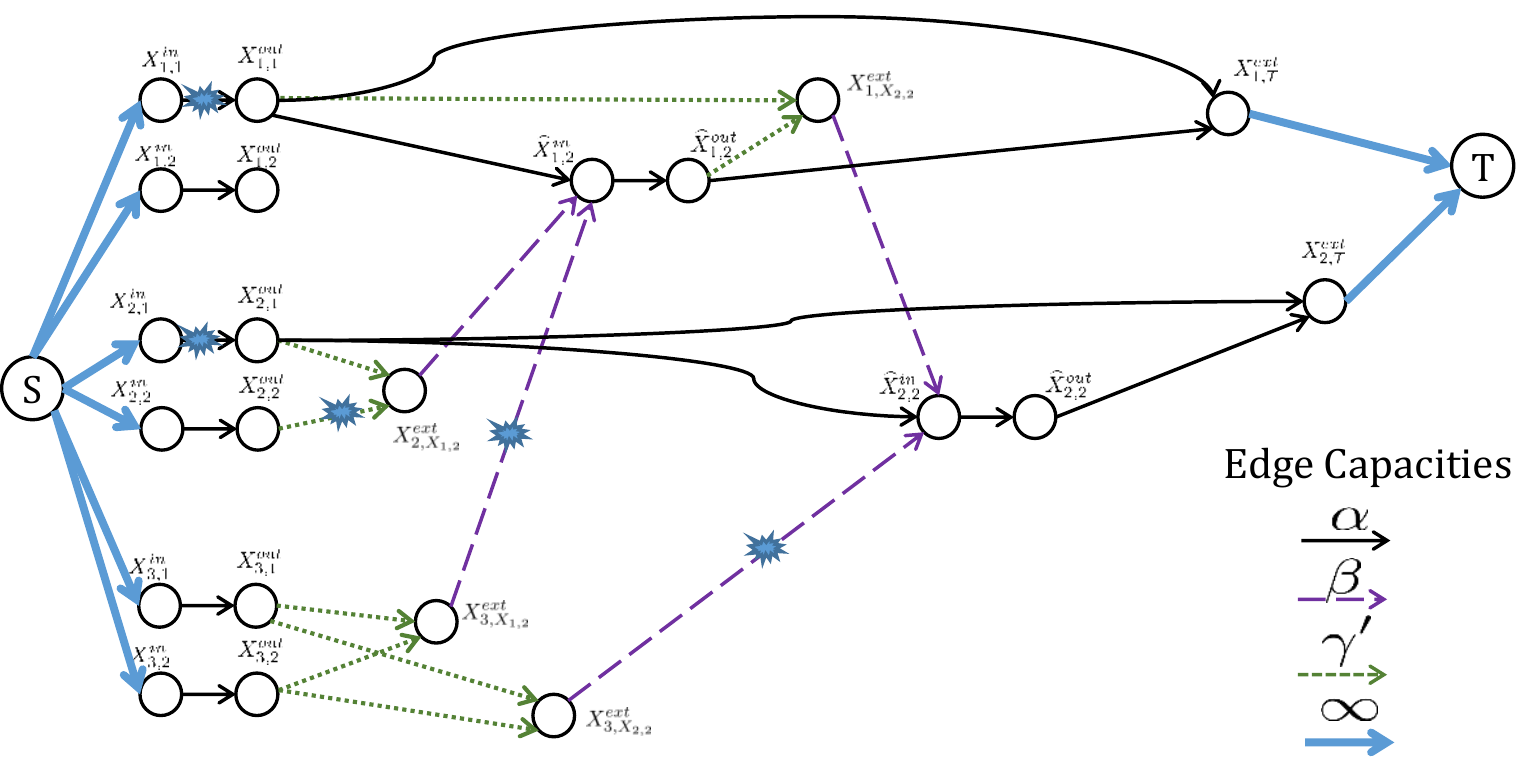}
	\caption{An illustration of the IFG used in cut-set based lower bound for $\gamma'$ in Theorem \ref{thm:helper_BW}. In this example, we assume $(n = 3, k = 2, d =3)(m = 2, \ell = 1)(\ell' = 2, \gamma = \alpha)$. The second node fails in clusters $1$ and $2$ in the respective order. Also indicated is our choice of the $S$-$T$ cut used in the bound derivation. }
	\label{fig:cut_gamma_prime}
	\vspace{-10pt}
\end{figure*}

\vspace{0.1in}

\begin{thm} \label{thm:helper_BW}
	For an optimal functional repair generalized regenerating code with parameters  $\{(n, k, d \geq k)$ $ (\alpha, \beta), (m, \ell)\}, \gamma = \alpha, \ell' = m$, the remote helper-node repair bandwidth $\gamma'$ is lower-bounded by
	\begin{eqnarray}
		\gamma' & \geq & \beta/(m-\ell), \label{eq:g'bound} \\
	\end{eqnarray}
	whenever $\alpha\geq(d-k+2)\beta$.
\end{thm}
\begin{proof}
	We consider data collection from clusters $1$ to $k$. Before data collection, the system experiences $k(m - \ell)$ repairs. Nodes $\ell+1, \ldots, m$ fail and get repaired in cluster $1$ in this respective order. This is followed by failure and repair of nodes  $\ell+1, \ldots, m$ in cluster $2$, and so on, until we consider failure and repair of nodes  $\ell+1, \ldots, m$ in cluster $k$. In terms of physical nodes, it may be noted that this is the same sequence of failures that was considered in the proof of Theorem \ref{thm:file_size}; however, in here, we will impose additional restrictions on the choice of the remote helper clusters. In this proof, external help is taken from the set of the first $d+1$ clusters, excluding the cluster where the failed node resides. Thus, for the repair of $X_{i, j}$, the indices of remote helper clusters are $\{1, \ldots, i-1, i+1, \ldots, k, k+1, \ldots, d+1\}$. The choice of  local helper nodes remain same as in the proof of Theorem \ref{thm:file_size}, where we used the first $\ell$ nodes in the cluster. An illustration of the IFG is shown in \Cref{fig:cut_gamma_prime}.
	
	It can be seen that the following cut-set separates the source from the data collector:
	\begin{itemize}
		\item $\{(X_{i,j}^{in} \to X_{i,j}^{out}), i\in [k], j\in [\ell]\}$. Total capacity of these edges is $k\ell\alpha$.
		\item For each $i, 1 \leq i \leq k$, the edge set with smaller capacity out of $A_1(i) \cup A_2(i)$ and $A_3(i)$ where
		\begin{itemize}
			\item $A_1(i) \triangleq \{(X_{i'}^{ext} \to \widehat{X}_{i,j}^{in}), i'\in [k+1,d+1],  j\in[\ell+1,m]\}$. Total capacity of edges in $A_1(i)$ is $(d-k+1)(m-\ell)\beta$. Recall that $(\widehat{X}_{i,j}^{in}, \widehat{X}_{i,j}^{out})$ is simply the replacement node for the failed node  $({X}_{i,j}^{in}, {X}_{i,j}^{out})$,  in the second IFG model that is used here.
			\item $A_2(i) \triangleq \{(X_{i',j'}^{out} \to X_{i',X_{i,j}}^{ext}), j \in [\ell+1,m],  i' \in [i+1,k], j' \in [\ell+1,m]\}$. Total capacity of edges in $A_2(i)$ is $(m-\ell)(k-i)(m-\ell)\gamma'$
			\item $A_3(i) \triangleq \{(\widehat{X}_{i,j}^{in}\to  \widehat{X}_{i,j}^{out}), j\in [\ell+1,m]\}$. Total capacity of edges in $A_3(i)$ is $(m-\ell)\alpha$.
		\end{itemize}
	\end{itemize}
	The capacity of the cut-set is given by
	\begin{eqnarray}
		C_{cut}& = & k\ell\alpha \ + \ (m-\ell)\sum_{i=1}^{k} \min\{\alpha, (d-k+1)\beta+(k-i)(m-\ell)\gamma'\}.
	\end{eqnarray}
	Since we consider optimal codes, we have
	\begin{eqnarray}
		C_{cut} & \geq & B^* \ = \  k\ell\alpha + (m-\ell)\sum_{i=1}^{k} \min\{\alpha, (d-i+1)\beta\}.
	\end{eqnarray}
	In this case, since we assume that $\alpha\geq(d-k+2)\beta$, we claim that
	\begin{eqnarray} \label{eq:gammadash}
		\gamma' & \geq  &\beta/(m-\ell).
	\end{eqnarray}
	To see why \eqref{eq:gammadash} is true, if we suppose on the contrary that $(m - \ell)\gamma' < \beta$, we have
	\begin{eqnarray}
		(d-k+1)\beta+(k-i)(m-\ell)\gamma' & < & (d-i+1)\beta, \ 1 \leq i \leq k - 1.
	\end{eqnarray}
	The above equation implies that
	\begin{eqnarray}
		\min\{\alpha, (d-k+1)\beta+(k-i)(m-\ell)\gamma'\} & \leq & \min\{\alpha, (d-i+1)\beta\}, 1 \leq i \leq  k - 2  \label{eq:gammadash_0}\\
		\min\{\alpha, (d-k+1)\beta+(k-i)(m-\ell)\gamma'\} & < & \min\{\alpha, (d-i+1)\beta\}, i =  k - 1, \label{eq:gammadash_1}
	\end{eqnarray}
	where \eqref{eq:gammadash_1} follows from the assumption that $\alpha\geq(d-k+2)\beta$. Clearly, adding up the two corresponding sides (L.HS. and R.H.S.) of \eqref{eq:gammadash_0}  and \eqref{eq:gammadash_1}, and comparing them contradicts the fact that $C_{cut}  \geq  B^* $. Thus, it must be true that $\gamma'  \geq  \beta/(m-\ell)$, whenever $\alpha\geq(d-k+2)\beta$.
\end{proof}

\vspace{0.1in}

The following theorem establishes the necessary condition on $\ell'$ for optimal codes. The proof is along the lines of proof of Theorem \ref{thm:helper_BW}, and is omitted.

\vspace{0.1in}

\begin{thm} \label{thm:elldash}
	For an optimal functional repair GRC with parameters  $\{(n, k, d \geq k) $ $(\alpha, \beta), (m, \ell)\}, \gamma = \gamma' = \alpha$, whenever $\alpha\geq(d-k+2)\beta$, each remote helper cluster must necessarily access all the $m$ nodes in the cluster while generating the $\beta$ symbols; i.e., whenever $\alpha\geq(d-k+2)\beta$, we have $\ell' = m$ for an optimal functional repair GRC.
\end{thm}

\section{Security Under Passive Eavesdropping} \label{sec:security}

In this section, we analyze resilience of the clustered storage system against passive eavesdropping. Our model of clustered storage systems is in part motivated by the need to provide security against an eavesdropper who may gain access to a subset of the clusters. In this context, we extend result in Theorem \ref{thm:file_size} to settings that require security. Below, we first introduce the model for security, and then present the revised file bound. Optimal code construction for security can be provided along the  same lines of Construction \ref{constr:exact}.

\subsection{Passive Eavesdropper Model}

The security model is along the lines of the passive eavesdropper model considered in \cite{pawar_security}, where the authors study security under the classical regenerating code framework. An eavesdropper (say, Eve) gains access to the entire content of any subset of $e$ clusters, where $1 \leq e \leq k$. Eve also gets to observe all the helper data that gets downloaded for repair of any node in these $e$ clusters. Eve is passive in the sense that Eve does not change any stored or repair data.  The properties of data collection and disk repair remain same as in the case of no eavesdropper (see Section \ref{sec:sys_model}). In this model we $1)$ ignore the effects of intra-cluster bandwidth, and $2)$ restrict ourselves to the setting of deterministic exact repair codes. By deterministic exact repair code, we mean that the helper data for the repair of any node is uniquely determined given the indices of the failed node, local helper nodes and remote helper clusters. We avoid the possibility that the same set of helpers can pass two possible sets of helper data for the repair of the same node.

We wish to store a file such that Eve does not gain any information about it, by having eavesdropped into any subset $E$ of $e$ clusters. To be precise, let $\mathcal{F}^{(s)}$ denote the random variable corresponding to the data file that gets securely stored. We assume the file $\mathcal{F}^{(s)}$ to be uniformly distributed over  $\mathbb{F}_q^{B^{(s)}}$, and thus $B^{(s)}$ denotes the file-size. We wish to ensure that the mutual information $\mathcal{I}(\mathcal{F}^{(s)}; \text{data observed by Eve}) = 0$. Note that the data observed by Eve not only includes the content of the $E$ clusters, but also any inter-cluster helper data that is received toward the repair of nodes in these clusters. We shall write $\mathcal{C}^{(s)}_m$ to denote a secure generalized regenerating code, and its parameter set will be identified with $\{(n, k, d),(\alpha, \beta),(m, \ell),(e)$\}.

\subsection{File Size Under Exact Repair}

In this section, we obtain an upper bound on the file-size $B^{(s)}$ of the exact repair secure generalized regenerating code $\mathcal{C}^{(s)}_m$. To derive the bound, we use information theoretic techniques similar to those used  in \cite{rbt}, \cite{pawar_security}. We begin with some necessary notation. Let $Y_{i, j} \in \mathbb{F}_q^{\alpha}, 1 \leq i \leq n, 1 \leq j \leq m$ denote the content stored in node $j$ of cluster $i$. We write $\bf{Y}_i$ to denote $[Y_{i,1} \ldots Y_{i,m}], 1 \leq i \leq n$. The property of data collection requires that
\begin{eqnarray} \label{eq:data_collect}
	H\left(\mathcal{F}^{(s)} | \{{ \bf Y}_i, i \in S \}\right) & = & 0 \ \forall S \subset [n], |S| = k,
\end{eqnarray}
where $H(.)$ denotes the entropy function computed with respect to $\log q$. Next, consider the repair of node $j$ in cluster $i$. Let $\mathcal{H} \subset [n]\backslash\{i\}, |\mathcal{H}| = d$, and $\mathcal{L} \subset [m]\backslash\{j\}, |\mathcal{L}| = \ell$ respectively denote the indices of remote helper clusters and local helper nodes that aid in the repair process. Let $Z_{i', i, j}^{\mathcal{H}, \mathcal{L}}$ denote helper data passed by cluster $i'$. Recall our assumption that the exact repair code is deterministic, and thus $Z_{i', i, j}^{\mathcal{H}, \mathcal{L}}$ is uniquely determined as a function of $(i, j), \{Y_{i, j'}, j' \in \mathcal{L}\}$ and $\{Y_{i'}, i' \in \mathcal{H}\}$. The property of exact repair is jointly characterized by the following set of inequalities:
\begin{eqnarray}
	H\left(Z_{i', i, j}^{\mathcal{H}, \mathcal{L}} | \bf{Y}_{i'} \right) & = & 0, \label{eq:exact_prop1}\\
	H\left(Z_{i', i, j}^{\mathcal{H}, \mathcal{L}}\right) & \leq & \beta, \label{eq:exact_prop2} \\
	H\left(Y_{i, j}| \{Z_{i', i, j}^{\mathcal{H}, \mathcal{L}}, Y_{i, j'}, i' \in \mathcal{H},  j' \in \mathcal{L}\}\right)& = & 0, \ \forall \mathcal{H}  \subset [n]\backslash \{i'\}, |\mathcal{H}| = d, \forall \mathcal{L} \subset [m]\backslash \{j'\}, |\mathcal{H}| = \ell. \label{eq:exact_rep}
\end{eqnarray}
Next, define $U_i$ to denote the collection of all the inter-cluster helper data ever received toward the repair of nodes in cluster $i$, i.e.,
\begin{eqnarray}
	U_i & = & \{Z_{i', i, j}^{\mathcal{H}, \mathcal{L}}, \ \text{for all choices of } i', j, \mathcal{H}, \mathcal{L}\}.
\end{eqnarray}
The property of being secure against the passive eavesdropper Eve is equivalent to saying that
\begin{eqnarray}
	\mathcal{I}(\mathcal{F}^{(s)}; \{{\bf Y}_i, U_i, \ i \in E\}) & = & 0, \ \forall E \subset [n], |E| = e. \label{eq:secure_prop}
\end{eqnarray}

The following theorem characterizes the file-size bound under passive eavesdropping.

\vspace{0.1in}

\begin{thm} \label{thm:secure}
	The file-size of a secure exact-repair deterministic generalized regenerating code having parameters $\{(n, k,d),(\alpha, \beta),(m, \ell), (e)\}$ is upper bounded by
	\begin{eqnarray}\label{eq:file_size_secure_rep}
		B^{(s)} & \leq & \ell (k-e)\alpha + (m - \ell)\sum_{i=e}^{k-1} \min \{\alpha, (d-i)^+\beta \}.
	\end{eqnarray}
\end{thm}
\begin{proof}
	Without loss of generality, let us assume that $E = \{1, 2, \ldots, e\}$. Using \eqref{eq:secure_prop}, the file-size $B^{(s)}$ is given by
	\begin{eqnarray}
		B^{(s)} & = & H(\mathcal{F}^{(s)}) \ = \ H(\mathcal{F}^{(s)}| \{{\bf Y}_i, U_i, \ i \in E\}) \ \leq \ H(\mathcal{F}^{(s)}| \{{\bf Y}_i \ i \in E\}). \label{eq:proof_sec_1}
	\end{eqnarray}
	From \eqref{eq:data_collect}, we further get that
	\begin{eqnarray} \label{eq:proof_sec_2}
		H(\mathcal{F}^{(s)}| \{{\bf Y}_i \ i \in E\}) & \leq & H(\{{\bf Y}_i, \ i \in \{e+1, \ldots, k\}\}| \{{\bf Y}_i \ i \in E\}).
	\end{eqnarray}
	Combining \eqref{eq:proof_sec_1} and \eqref{eq:proof_sec_2}, we get
	\begin{eqnarray}
		B^{(s)} & \leq & H(\{{\bf Y}_i, \ i \in \{e+1, \ldots, k\}\}| \{{\bf Y}_i \ i \in E\}) \\
		& = & \sum_{i = e+1}^{k}H(\bf{Y}_i | \bf{Y}_1, \ldots, \bf{Y}_{i-1}) \\
		& = & \sum_{i = e+1}^{k}\sum_{j = 1}^{m}H(Y_{i,j} | Y_{i,1}, \ldots, Y_{i,j-1}, \bf{Y}_1, \ldots, \bf{Y}_{i-1})\\
		& = & \sum_{i = e+1}^{k}\sum_{j = 1}^{\ell}H(Y_{i,j} | Y_{i,1}, \ldots, Y_{i,j-1}, {\bf Y}_1, \ldots, {\bf Y}_{i-1}) + \sum_{i = e+1}^{k}\sum_{j = \ell+1}^{m}H(Y_{i,j} | Y_{i,1}, \ldots, Y_{i,j-1}, {\bf Y}_1, \ldots, {\bf Y}_{i-1}) \nonumber \\
		& & \label{eq:exact_pf_0} \\
		& \leq & \ell(k-e) \alpha + \sum_{i = e+1}^{k}\sum_{j = \ell+1}^{m}H(Y_{i,j} | Y_{i,1}, \ldots, Y_{i,j-1}, {\bf Y}_1, \ldots, {\bf Y}_{i-1}) \label{eq:exact_pf_1} \\
		& \leq & \ell(k-e) \alpha + \sum_{i = e+1}^{k}\sum_{j = \ell+1}^{m}\min(\alpha, (d-(i-1))^+\beta)) \label{eq:exact_pf_2} \\
		& = & \ell(k-e) \alpha + (m - \ell)\sum_{i=e}^{k-1} \min \{\alpha, (d-i)^+\beta \}, \label{eq:exact_pf_01}
	\end{eqnarray}
	where \eqref{eq:exact_pf_2} follows from \eqref{eq:exact_prop1}-\eqref{eq:exact_rep}. This completes the proof of the upper bound.
\end{proof}

\vspace{0.1in}

As before, whenever $d > 0$, we associate the operating point $\alpha = d\beta$ with the minimum repair bandwidth secure regenerating codes. Construction \ref{constr:exact} can be easily adapted to construct optimal secure exact repair codes at the MBR point. In the modified construction,  one combines $\ell$ secure MDS codes for the wiretap-II channel~\cite{wiretap}\cite{arunkumar_secure}, and $(m-\ell)$ classical exact-repair secure MBR codes~\cite{pawar_security}\cite{rashmi_secure}. The construction and proof of optimality are similar to Construction \ref{constr:exact}; we avoid a full description here. A pictorial illustration of the secure code construction appears in Fig. \ref{fig:code_secure}.  Finally, we note that the bound in \eqref{eq:exact_pf_01} is weak at the MSR point, since it is known that classical exact-repair secure MSR codes in general cannot  achieve file size $\sum_{i=e}^{k-1} \min \{\alpha, (d-i)^+\beta \}=(k-e)\alpha$ \cite{sasidharan2014optimality,rawat2016note}. It is an interesting question as to whether the analysis like in \cite{sasidharan2014optimality,rawat2016note} can be used to provide a tighter bound for secure MSR codes under the framework of GRCs considered in this paper.

\begin{figure*}
	\centering
	\includegraphics[height=2.5in]{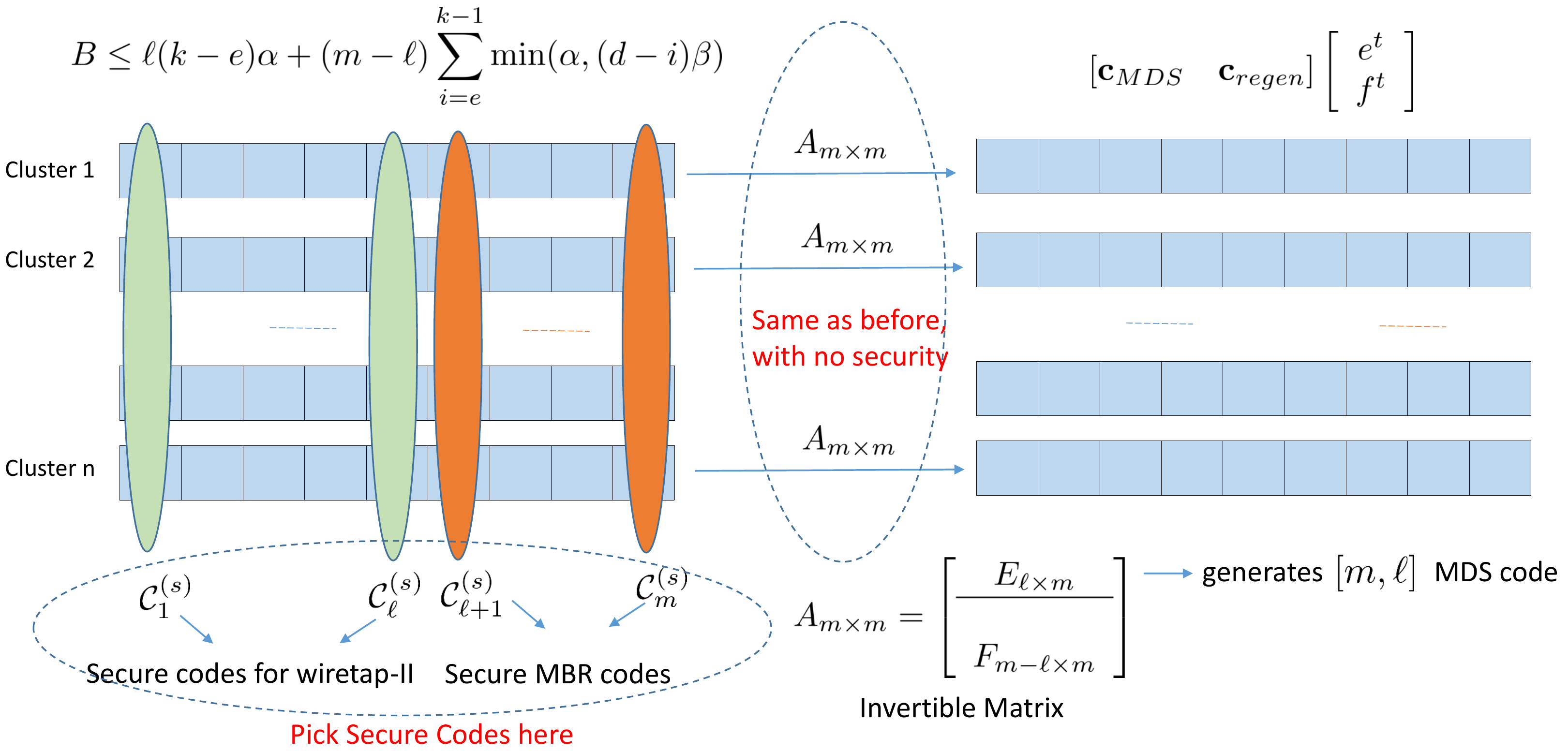}
	\caption{Illustration of the exact repair secure code construction. We first stack $\ell$ secure-MDS codes for the wiretap-II channel, and $(m - \ell)$ classical secure MBR codes, and then transform each row via the invertible matrix $A$. The first $\ell$ rows of the matrix $A$ generate an $[m, \ell]$ MDS code.}
	\label{fig:code_secure}
\end{figure*}

\section{Conclusions} \label{sec:con}

To conclude, we study the problem of storage-overhead vs repair-bandwidth overhead in clustered storage systems. The notion of clustering is used in both data collection and node repair. For data collection, we require retrievability using content from any set of $k$ clusters. For node repair, we take the help of surviving local nodes in the host cluster, as well as from other remote clusters.
We first characterize the optimal file-size that is achievable while ignoring intra-cluster bandwidth costs, and then obtain bounds on intra-cluster bandwidth costs that is needed to achieve this file-size. Our results show that while it is beneficial to increase the number of local helper nodes ($\ell$) during repair in order to simultaneously improve both storage and inter-cluster bandwidth costs, increasing $\ell$ has an adverse effect on intra-cluster repair bandwidth. Our bounds on file-size and intra-cluster bandwidth give guidelines for choosing the desired number of local helper nodes in practice, based on the relative costs of the various metrics. We present constructions of optimal exact and functional repair codes, which enable operating points for clustered systems that are not achievable via previously known coding solutions. We also analyze the resilience of the system against passive eavesdropping.

Two key questions remain at the end of this work. Firstly, our bounds on intra-cluster bandwidth are derived under the assumption of functional repair. It is unclear if these bounds hold under exact repair; specifically at the minimum-inter-cluster bandwidth operating point. The exact repair constructions in this paper, although they have optimal file-size (and inter-cluster bandwidth), incur the maximum possible intra-cluster bandwidth. Secondly, the bound on any one of the intra-cluster bandwidth related parameters (say, $\gamma$) was derived without limiting the other two other parameters ($\gamma', \ell'$). It is of special interest to know the simultaneously optimality of the the bounds in \eqref{eq:gamma} and \eqref{eq:intra_bound_helper2}. We believe that a first step in this direction would be to prove a converse statement (achievability) to \eqref{eq:intra_bound_helper2}. Achievability of \eqref{eq:intra_bound_helper2} is indeed suggested by RLNC-based simulations.

Finally, model extensions of interest include the case of simultaneous recovery from multiple node failures in a given cluster. While studying recovery for multiple node failures, it is of interest to consider presence of local parity relations in each cluster, even when $d > 0$.

\bibliographystyle{IEEEtranN}
\bibliography{citations}

\end{document}